\documentclass[sigconf, nonacm, authorversion]{aamas}

\usepackage{balance} % for balancing columns on the final page
\usepackage{graphicx}  % for \rotatebox
\usepackage{tikz}
\usepackage{amsfonts}
\usepackage{amsmath}
\usepackage{amsthm}
\usepackage{thm-restate}
\usetikzlibrary{shapes.misc, fit, decorations.pathreplacing,calligraphy,positioning,calc} %for crossing-out arrows and rectangles around nodes

\usepackage{mathrsfs}
\usepackage{macros}
\newcommand{\pentacle}{%
  \mathbin{%
    \begin{tikzpicture}[scale=0.15, baseline=-0.5ex]
      \begin{scope}[rotate=180]
        \draw[line width=0.8pt] (0,0) circle(1cm);
        \foreach \i in {0,72,...,288} {
          \coordinate (P\i) at ({cos(\i+18)}, {sin(\i+18)});
        }
        \draw[line width=0.4pt] (P0) -- (P144) -- (P288) -- (P72) -- (P216) -- cycle;
      \end{scope}
    \end{tikzpicture}%
  }%
}

\theoremstyle{definition}
\newtheorem{definition}{Definition}

\newtheorem{theorem}{Theorem}

\newtheorem{example}{Example}
\newtheorem{remark}{Remark}

\usepackage{algpseudocode}
\usepackage{algorithm}

\algnewcommand\algorithmiccase{\textbf{case}}
\algdef{SE}[CASE]{Case}{EndCase}[1]{\algorithmiccase\ #1}{\algorithmicend\ \algorithmiccase}%
\algtext*{EndCase}
\algtext*{EndIf}
\algtext*{EndFor}

\makeatletter
\newenvironment{breakablealgorithm}
  {% \begin{breakablealgorithm}
   \begin{center}
     \refstepcounter{algorithm}% New algorithm
     \hrule height.8pt depth0pt \kern2pt% \@fs@pre for \@fs@ruled
     \renewcommand{\caption}[2][\relax]{% Make a new \caption
       {\raggedright\textbf{\fname@algorithm~\thealgorithm} ##2\par}%
       \ifx\relax##1\relax % #1 is \relax
         \addcontentsline{loa}{algorithm}{\protect\numberline{\thealgorithm}##2}%
       \else % #1 is not \relax
         \addcontentsline{loa}{algorithm}{\protect\numberline{\thealgorithm}##1}%
       \fi
       \kern2pt\hrule\kern2pt
     }
  }{% \end{breakablealgorithm}
     \kern2pt\hrule\relax% \@fs@post for \@fs@ruled
   \end{center}
  }
\makeatother

\allowdisplaybreaks

\newcommand{\rust}[1]{{\color{blue}#1}}

\newif\ifdraft\drafttrue%false

\ifdraft
\newcommand{\munyque}[1]{{\color{cyan}MM: #1}}
\newcommand{\munyquem}[1]{{\color{cyan}#1}}

\newcommand{\davide}[1]{{\color{red}DC: #1}}

\else
\newcommand{\munyque}[1]{}
\newcommand{\munyquem}[1]{{#1}}
\newcommand{\davide}[1]{}
\fi

\setcopyright{ifaamas}
\acmConference[AAMAS '26]{Proc.\@ of the 25th International Conference
on Autonomous Agents and Multiagent Systems (AAMAS 2026)}{May 25 -- 29, 2026}
{Paphos, Cyprus}{C.~Amato, L.~Dennis, V.~Mascardi, J.~Thangarajah (eds.)}
\copyrightyear{2026}
\acmYear{2026}
\acmDOI{}
\acmPrice{}
\acmISBN{}

\acmSubmissionID{713}

\title[On Angels and Demons]{On Angels and Demons:\\ Strategic (De)Construction of Dynamic Models} 

\thanks{This is an extended version of the paper with the same title that will appear in AAMAS 2026,
which contains technical appendices with proof details.}
\author{Davide Catta}
\affiliation{
  \institution{LIPN, CNRS, Université Sorbonne Paris Nord}
  \city{Villetenause}
  \country{France}}
\email{catta@lipn.univ-paris13.fr}

\author{Rustam Galimullin}
\affiliation{
  \institution{University of Bergen}
  \city{Bergen}
  \country{Norway}}
\email{rustam.galimullin@uib.no}

\author{Munyque Mittelmann}
\affiliation{
  \institution{LIPN, CNRS, Université Sorbonne Paris Nord}
  \city{Villetenause}
  \country{France}}
\email{mittelmann@lipn.univ-paris13.fr}

\begin{abstract}
In recent years, there has been growing interest in logics that formalise strategic reasoning about agents capable of modifying the structure of a given model.  This line of research has been motivated by applications where a modelled system evolves over time, such as communication networks, security protocols, and multi-agent planning. In this paper, we introduce three logics for reasoning about strategies that modify the topology of weighted graphs. In \textit{Strategic Deconstruction Logic}, a destructive agent (the demon) removes edges up to a certain cost. In \textit{Strategic Construction Logic}, a constructive agent (the angel) adds edges within a cost bound. Finally, \textit{Strategic Update Logic} combines both agents, who may cooperate or compete. We study the expressive power of these logics and the complexity of their model checking problems.
\end{abstract}


\newcommand{\BibTeX}{\rm B\kern-.05em{\sc i\kern-.025em b}\kern-.08em\TeX}

\begin{document}

%%% The following commands remove the headers in your paper. For final 
%%% papers, these will be inserted during the pagination process.

\pagestyle{fancy}
\fancyhead{}

%%% The next command prints the information defined in the preamble.

\maketitle 

%%%%%%%%%%%%%%%%%%%%%%%%%%%%%%%%%%%%%%%%%%%%%%%%%%%%%%%%%%%%%%%%%%%%%%%%

%\munyque{
%Add a subtitle to clarify/make it easy to find
% }

%\rust{How about the current title? Also, I was thinking about changing SDL to Strategic Deconstruction Logic, SCL to Strategic Construction Logic, and keep SUL as it is. In such a way we go away from an assumption that we merely have a variant of sabotage logic. What do you guys think?}

%\munyque{I like it. I just wonder if it is enough to bring us the KR reviewers (perhaps we can add "logics" in the keywords). }

\section{Introduction} 

The growing adoption and reliance on autonomous systems in safety-critical applications, such as autonomous vehicles \cite{bila2016vehicles} and cybersecurity systems \cite{bhamare2020cybersecurity}, call for reliable verification methods. 
%Model checking \cite{HandbookMC2018} is an approach based on formal methods for verifying whether a system satisfies desirable temporal properties, such as safety and reachability requirements. 
Model checking \cite{HandbookMC2018} is one of the de facto standard approaches to verification of such temporal properties of a given system as safety and reachability. 
In the context of Multi-Agent Systems (MAS), this approach was extended to capture the interaction %(whether cooperative, adversarial, or hybrid) 
and strategic behaviour of autonomous agents. %,  which are typically formalized using strategy logics. 
Properties of a system that one would like to verify are usually expressed in temporal or strategic logics, like the Linear Temporal Logic (LTL) \cite{Pnu77}, Computation Tree Logic (CTL) \cite{ctl}, and Alternating-time Temporal Logic (ATL) \cite{alur2002}.

In model checking, a given system is typically represented using a \textit{static} model (e.g., labelled state-transition models or concurrent game structures), which describes a fixed set of configurations, or states, of the system and how it transitions between them. Such models are incapable of capturing scenarios in which the structure of the system may dynamically change, for instance, caused by actions of agents or the removal of vulnerable components. 
The assumption of static models limits the applicability of model checking approaches, as many real-life applications are inherently \textit{dynamic}.
One example would be the addition of a new route in a metro system, which may cause bottlenecks at interchange stations. %An additional instance 
Another example is the implementation of defence strategies to prevent cyberattacks exploiting system vulnerabilities \cite{CattaLM23}.
%Incorporating model dynamics into system analysis is vital for %achieving 
%meaningful and reliable verification of dynamic systems. %in this type of system.

Motivated by these limitations, there has been a growing interest in logic-based approaches to the specification and verification of dynamic systems in recent years. 
A notable line of research draws inspiration from the sabotage game \cite{Benthem05}, which is a reachability game on graphs where one player, called the demon, can delete edges to obstruct the other player, the traveller, from reaching her goal. \textit{Obstruction Logic} (OL) \cite{CattaLM23} was recently proposed to analyse sabotage-like games on weighted graphs, where the demon can temporarily disable edges whose weights do not exceed a specified cost. OL, however, only captures a particular type of graph modifications, where removed edges are immediately restored after the traveller's move. Returning to the cybersecurity example, while it can represent the existence of defensive measures that briefly block an attacker's access to a sensitive module, it fails to capture measures such as removing vulnerable execution paths. 

\textbf{Our contribution.}
We propose three novel logics for reasoning about strategic permanent change of weighted graphs. 
The first one, \textit{Strategic Deconstruction Logic} (SDL), has a similar flavour to OL:  a destructive agent, or the demon, permanently removes edges up to a certain cost. %Such a permanent removal of edges allows us to reason about the demon removing edges long before the traveller reaches them, which is impossible in OL.  
In the cybersecurity setting, permanent removal of edges allows blocking vulnerable paths long before the attacker reaches them, which is impossible in OL.  
The second logic, \textit{Strategic Construction Logic} (SCL),
considers, instead, a constructive agent, or the angel, that can add new edges within a cost bound. Finally, \textit{Strategic Update Logic} (SUL) combines both agents acting concurrently à la ATL. In SUL, the angel and demon can cooperate towards the same goal, and their behaviour may even include joint strategies to mimic OL-like strategies. 
Additionally, SUL captures situations in which they are adversarial, thereby enhancing the strategic dimension of the problem. The main advantage of the proposed logics is that they enable the modelling and verification of dynamic systems in which access control and defence mechanisms can be employed during the execution of the system. % with potential application to  cybersecurity and workflow management.

%We provide several novel technical contributions related to the proposed logics. F
For the new logics, we, first, demonstrate that each one of them is strictly more expressive than CTL. Second, we show that SDL and SCL are, expressivity-wise, incomparable, and that SUL subsumes both of them. 
Next, we study and discuss their nuanced relation with OL. 
We also investigate the model checking problem for the logics and show that for SDL, SCL, and a fragment of SUL, the problem is PSPACE-complete, whereas for the full SUL it is in EXPSPACE. %The upper and lower bounds are obtained via alternating algorithms and reductions from the QBF problem, respectively.

\section{Logics of Angels and Demons}
\label{sec:logics}
%We introduce logics for  the strategies of a particular class of two-player games played on a weighted, directed graph where the transition relation is serial. The game involves a modifier and a traveler (or attacker). There are two mainly variants of these games that we now describe

%\paragraph{Demonic variant }We first consider games in which the modifier is called the demon. Each turn of the game proceeds as follows, given a natural number $n$ and a node $s$:

%\begin{itemize}
 %   \item The demon removes a subset of edges from the graph (possibly empty) whose total weight does not encompass $n$, taking care not to violate the seriality of the transition relation.
 %   \item The traveler then moves from $s$ to an accessible node.
%\end{itemize}

%The game continues in this manner. The demon wins if the infinite sequence of states generated by the traveler’s choices satisfies a given temporal property.

%\paragraph{Angelic Variant} The Angelic variant can be described in the same terms as the Demonic one, the difference being that the Angel adds a subset of edges at each turn.

We introduce logics for games played on weighted directed graphs with a serial transition relation. The games involve players who can modify the graph and \textit{the traveller} who moves along the edges of the graph. In the first case, the modifying player is called \textit{the demon}, and she can remove a subset of edges from the graph, possibly empty, up to a certain cost. Then the traveller makes a move in the modified graph. In the second case, the modifying player is called \textit{the angel}\footnote{The name of the edge-removing player, the demon, comes from the literature on sabotage games \cite{aucher2018modal}.  Hence, as the counterpart to the demon, we call the edge-adding agent the angel. %These names come with their own connotations implying that the demon somehow is against the traveller, and the angel's goal is to help her. However, neither demon nor angel cooperate with the attacker or the traveller. 
To alleviate the connotations coming with such names, we could have called the demon the \textit{remover}, or \textit{saboteur}, and the angel the \textit{constructor}. Indeed, in the setting of adding and removing edges, the demon can be viewed as benevolent if she removes edges to bad or undesirable states, and the angel can be considered malevolent if she adds edges to undesirable states. 
Perhaps, to play on the metaphor, we could view actions of celestial beings as giving or taking away options, and it is up to a mortal, the traveller, to make her own choices. }, and she is able to add edges to the graph up to a certain cost. Finally, we consider the scenario where both the demon and the angel are present and they make concurrent choices of which edges to add and to remove. In all of the logics, we are able to reason about strategies of the modifying players reaching a temporal goal regardless of the moves of the traveller. 

\subsection{Models}

%If $u$ is a countably infinite sequence of elements from some set $S$, we denote by $u(i)$ its $ith+1$ element and by $u(\leq i)$ its finite prefix $u(0),\ldots, u(i)$. 

%For the rest of the paper, we fix a countable set $\At$ of atomic propositions (or atoms) that will be denoted by small-case letters from the end of the alphabet ($p,q,r\ldots$), possibly indexed. 

Let  $\At = \{p,q,r...\}$ be a countable set of \emph{atoms}. %(or atoms).
%A model is a directed graph with a serial edge relation in which states are labeled by atoms, and a positive natural number, a cost, is assigned to any pair of states.

\begin{definition}[Model]
    A \emph{model} is a tuple $\M=(S,\to,\mathcal{V}, \mathcal{C})$,
    where: 
    
    \begin{itemize}
        \item  $S$ is a non-empty set of states.
        \item $\to \subseteq S\times S$ is a serial %a total (or serial) 
        binary relation over $S$. We will write $s\to s'$ for $(s,s')\in \to$. %and say that $s'$ is a successor of $s$ or that $s$ is a predecessor of $s'$. 
        \item $\mathcal{V}: \At \to 2^S$ is a valuation function specifying which atoms hold in which states.
        %function sending each atomic proposition to a subset of $S$. States belonging to $\mathcal{V}(p)$ represent those in which the atomic proposition $p$ holds true. 
        \item $\mathcal{C}: S\times S \to \mathbb{N}^+$ is a cost function assigning to \emph{any} pair of states a positive natural number. Intuitively, this number represents the cost of removing or adding an edge between two given states. %Intuitively, this number represents the cost for one agent to act (either delete or add an arc) between the considered states. 
    \end{itemize}
Given a model $\M$, its set of states will be denoted by $S^\M$, and its set of edges by $\xrightarrow{\M}$. 
A \emph{pointed model} is a pair $(\M,s)$, where $\M$ is a model and $s \in S^\M$. %  is a state. 
%If $\M$ is a model and $A\subseteq \xrightarrow{\M}$, we write $\M\setminus A$ as a shorthand for $(S,\to \setminus A , \mathcal{V}, \mathcal{C})$. 
We write $\M\setminus A$ for $(S,{\to \!\! \setminus A} , \mathcal{V}, \mathcal{C})$, where $A\subseteq \xrightarrow{\M}$. And we write $\M \cup A$ for $(S, \to \!\! \cup A, \mathcal{V}, \mathcal{C})$, where $A\subseteq S^\M \times S^\M$.

Given a model $\M$, we let $True(s) = \{p \in \At | s \in \mathcal{V}(p)\}$ be the set of atoms true in state $s$. Then, we define the \emph{size} of $\M$ as $|\M| = |S| + |\!\!\to\!\!| + \sum_{s \in S} |True (s)| + \sum_{s, t \in S} \mathcal{C}(s,t)$, where integers are encoded in binary. Finally, we call $\M$ \emph{finite}, if $|\M|$ is finite. 
\end{definition}

\begin{remark}
Note that we have defined a single cost for both removing and adding an edge. We could have relaxed this assumption, and all the results in the paper would hold. We stick with the current definition for simplicity.
\end{remark}

\def\sacc#1{\overset{#1}{\rightsquigarrow}}
\def\racc#1{\overset{#1}{\triangleright}}

\begin{definition}[(De)Construction and Updates]

Given two pointed models $(\M,s)$ and $(\M',s')$, we say that $(\M',s')$ is \emph{deconstruction accessible} from $(\M,s)$ with cost at most $n$, denoted by $(\M,s) \sacc{n} (\M',s')$, iff $\M'=\M \setminus A$ for some $A\subset \xrightarrow{\M}$, $s\xrightarrow{\M'} s'$, and $\mathcal{C}(A) \leqslant n$, where $\mathcal{C}(A)$ is $\sum_{x\in A} \mathcal{C}(x)$. We will call $\M'$ an \emph{$n$-submodel} of $\M$. %In the context of sabotage accessibility, we will call the agent who chooses an $n$-submodel the \textit{demon}, and the agent who makes the move $s\xrightarrow{\M'} s'$ is an \emph{attacker}.

We say that $(\M',s')$ is \textit{construction accessible} from $(\M,s)  $ with cost at most $n$, denoted $(\M,s) \racc n (\M',s')$, if and only if $\M'=\M \cup A$ for some $A\subseteq ((S^\M \times S^\M)\setminus \to^\M) $, $\C(A) \leqslant n$, and $s\xrightarrow{\M'} s'$. We will also call $\M'$ an \textit{$n$-supermodel} of $\M$. %The agent who chooses an $n$-supermodel is called the \textit{angel}, and the agent who makes the move $s\xrightarrow{\M'} s$ in the context of rebuild accessibility will be called a \emph{traveller}.

Given a model $\M$, its $n$-supermodel $\M_1$ obtained by adding the set of edges $A$, and $m$-submodel $\M_2$ obtained by removing the set of edges $B$, we will denote by $\M_1 \star \M_2 = (\M \setminus B) \cup A$ the resulting model after edges $B$ were removed and edges $A$ were added. We will call  $\M_1 \star \M_2$ an \textit{$n$-$m$-update}. Note that $\M_1 \star \M_2$ is well-defined as the sets of edges $A$ and $B$ are disjoint. %that the angel and the demon manipulate are disjoint.

For two pointed models $(\M, s)$ and $(\M', s')$, we say that $(\M', s')$ is \textit{update accessible} from $(\M, s)$ with costs $n$ and $m$, denoted by $(\M, s) \overset{n,m}{\Rightarrow}(\M', s')$, iff $\M'$ is an $n$-$m$-update of $\M$ and $s \xrightarrow{\M'} s'$. 

The agent making the move $s \xrightarrow{\M'} s'$ in any of the three contexts of deconstruction, construction, or update accessibility,  is called the \textit{traveller}.
\end{definition}

Intuitively, a model $\M'$ is deconstruction accessible with cost at most $n$ from $\M$  if we can obtain $\M'$ by removing edges from $\M$ with the total cost of up to $n$. The $n$-supermodel $\M'$ is obtained from $\M$ by adding a \textit{new} set of edges that are not already present in $\M$ and whose total cost does not exceed $n$.
Observe that since all the notions of accessibility above are defined between two (pointed) models, all the operations preserve seriality.

\subsection{Strategic Deconstruction Logic}

We call an agent who is able to remove edges \textit{the demon}. As opposed to sabotage games \cite{Benthem05}, we are interested in not merely one-step actions of the demon, but rather in \textit{strategies} of the demon to ensure some property against all moves of the traveller. 

\begin{definition}[Demonic Strategy]
\label{def:demstrat}
    Let $\pi$ be a (countably) infinite sequence of pointed models.  Then $\pi$ is a \emph{decreasing model path with cost $n$} iff for every $i \geqslant 0$ we have that $\pi(i) \sacc{n} \pi(i+1)$. Note that such a sequence is indeed infinite, as we can always take a $0$-submodel of the given model, i.e., keep the current model intact.
    We will denote by $\pi^\M$ the corresponding sequence of (non-pointed) models obtained by dropping the state from each element of $\pi$.  

A \emph{demonic strategy} is a function $\strat$ that, given as an input a pointed model $(\M,s)$, outputs a set of edges $A\subset \xrightarrow{\M}$.

 A decreasing model path $\pi $ is \emph{compatible} with a demonic strategy $\strat$ iff for all $i \in \mathbb{N} $, we have that $\strat( \pi(i))=A$ implies $\pi^{\M}(i+1)= \pi^\M (i) \setminus A$. We let $Out(\strat,(\M,s))$ denote the set of paths that are compatible with $\strat$ and whose first component is $(\M,s)$. A demonic strategy $\strat$ has cost $n$ (in this case, it will be called \emph{$n$-strategy} $\strat$) iff each model path that is compatible with the strategy has cost $n$. 
\end{definition}

\begin{definition}[Strategic Deconstruction Logic]
\label{def:sdl}
State $(\varphi)$ and path $(\psi)$ formulae are defined by mutual recursion: %via the following grammars: 
\begin{align*}\label{eq:formSDL}
\varphi := p \mid \neg \varphi \mid (\varphi \land \varphi) \mid \estrat{n} {\psi} \qquad
\psi := \nextt \varphi \mid \varphi \until \varphi \mid \varphi \release \varphi
\end{align*}

where $p\in \At$ and $n\in \mathbb{N}$. Formulae of \textit{Strategic Deconstruction Logic} (SDL) are all and only the state formulae. Constructs $\estrat{n} {\psi}$ are read as `there is a demonic $n$-strategy such that for all moves of the traveller, $\psi$ holds.' Temporal modalities $\nextt \varphi$ are read as `$\varphi$ is true in the ne$\mathsf{X}$t step', modalities $\varphi \until \psi$ mean `$\varphi$ holds $\mathsf{U}$ntil $\psi$ is true', and modalities $\varphi \release \psi$ are read as `truth of $\varphi$ $\mathsf{R}$eleases the requirement for the truth of $\psi$'.

We define the operator for \textit{sometime} as $\mathsf{F} \varphi := \top \until \varphi$, and \textit{always} as $\mathsf{G} \varphi := \bot \release \varphi$. The dual of the demonic operator is defined as $\astrat n  \psi: = \neg\estrat{n} \neg \psi$, and is read as `for all demonic $n$-strategies, there is a move of the traveller such that $\psi$ holds.'
Other propositional connectives, like implication, are defined as usual, and the conventions for removing parentheses hold.  We will sometimes call $n$ in strategic demonic operators $\estrat{n} {\psi}$ a \textit{resource bound}.

%\munyque{Note Out is defined analogously. }
Given a formula $\varphi$ of SDL, its \textit{size}, denoted by $|\varphi|$, is the number of symbols in $\varphi$ with integers encoded in binary.
\end{definition}

\begin{definition}[SDL Semantics]
\label{def::sdl_sem}
    Let $(\M,s)$ be a pointed model and $\varphi$ be a formula of SDL. We define the \emph{satisfaction relation} $(\M,s)\models \varphi$ by the induction on $\varphi$ omitting Boolean cases for brevity: 

$\begin{array}{l l l}
   %(\M,s) \models p  & \text{iff} &  s\in \mathcal{V}(p) 
   %\\
   
   %(\M,s)\models \neg \varphi  & \text{iff} & (\M,s)\not\models \varphi
   %\\
    %(\M,s)\models \varphi_1 \land \varphi_2 & \text{iff} & (\M,s)\models \varphi_1 \text{ and } (\M,s)\models \varphi_2 
   %\\
   (\M,s) \models \estrat n \psi & \text{iff} & \text{there is an $n$-strategy } \strat \text{ s.t. for all}
   \\
   &  & \pi\in Out(\strat, (\M,s)) \text{ we have } \pi\models \psi\\
   %   \rust{alt}(\M,s) \models \estrat n \psi & \iff & \exists \strat^n \forall \pi\in Out(\strat^n, (\M,s)) \text{ s.t. } \pi\models \psi\\
      
    %  \rust{alt}(\M,s) \models \astrat n \psi & \iff & \forall \strat^n \exists \pi\in Out(\strat^n, (\M,s)) \text{ s.t. } \pi\models \psi
   
\end{array}$

 %\noindent Let $\pi$ be a  model path and $\mathcal{C}$ a cost function. Given a path formula $\psi$, we define the satisfaction relation $\pi,\mathcal{C}\models \psi$ by induction on $\psi $:
 \noindent Given a path formula $\psi$ and a path $\pi$, the satisfaction relation $\pi \models \psi$ is defined by the induction on $\psi$:  

$\begin{array}{l l l}
   \pi \models \nextt \varphi_1  & \text{iff} &   \pi(1) \models \varphi_1  
   \\
   
%   \pi \models  \varphi_1 \until \varphi_2 & \text{iff}  & \text{there is } j\leq |\pi| \text{ s.t., } \pi(j)\models \varphi_2 
 %  \\
    \pi \models  \varphi_1 \until \varphi_2 & \text{iff}  & \text{there is } j\in \mathbb{N} \text{ such that } \pi(j)\models \varphi_2 
   \\
   & &  \text{and } \pi(i)\models \varphi_1 \text{ for each } 0\leqslant i < j
   \\
   % \pi \models \varphi_1 \release \varphi_2 & \text{iff}  & \text{either  }  \pi(j)\models \varphi_2 
   %\text{ for each } j\leq |\pi| \\
   \pi \models \varphi_1 \release \varphi_2 & \text{iff}  & \text{either  }  \pi(j)\models \varphi_2 
   \text{ for each } j \in \mathbb{N} \text{ or}\\
   & &  \text{there is } k \in \mathbb{N} \text{ such that }\pi(k)\models \varphi_1 
   \\
   && \text{and } \pi(i) \models \varphi_2 \text{ for all } 0 \leqslant i \leqslant k
\end{array}$
    
\end{definition}

Observe that we can define the standard modal box and diamond modalities using demonic strategies with 0 resources. Indeed, take $\square \varphi := \estrat 0 \nextt \varphi$ and $\Diamond \varphi:= \astrat 0 \nextt \varphi$. It is easy to verify that the box and diamond have exactly the intended semantics as the only strategy the demon can play is keeping the model intact.

%Consider the formula $\varphi= \estrat 1 \nextt \bot \land \estrat  $

\begin{example}{(Access control)}
\label{sec:example1}
    Let model $\M_{1}$ in Figure \ref{fig:example} depict a computational system managed by a security engineer. %where solid lines represents its transitions.
    The system user can be seen as the traveller, whereas the security engineer controls its access (i.e. the demon). %We assume that $\M_{1}$ contains the atomic propositions $error$, $admin$, and $server$.
    Also, let $error$, $admin$, and $server$ be atoms. 
State $s_0$ represents an authentication stage and leads to state $s_1$ if the user enters an incorrect password and to $s_2$ otherwise. Node $s_1$ is a failure state and prevents the user from accessing any other state of the system, including going back to $s_0$. States $s_2$ and $s_3$ are system modules within the server, with $s_3$ representing an administrative module. 
The valuation %function 
is as follows: $\mathcal{V}(s_0) =\emptyset$,  $\mathcal{V}(s_1) =\{error\}$, $\mathcal{V}(s_2) =\{ server\}$, and $\mathcal{V}(s_3) =\{server, admin\}$.

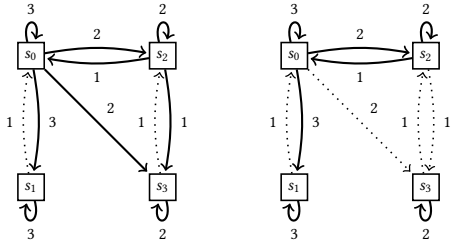
\begin{figure}[h!]
\centering
%\scalebox{0.8}{
 \begin{tikzpicture}[->,shorten >=1pt,auto,node distance=3cm and 3cm, semithick,minimum size=0.5cm, scale=0.7, transform shape]%,minimum size=1.cm

%\node[rectangle, draw=black, fill=white] (A1) {$s_0$};
%\node[rectangle, draw=black, fill=white] (B1) [below=of A1] {$s_1$};
%\node[rectangle, draw=black, fill=white] (C1) [right=of A1] {$s_2$};
%\node[rectangle, draw=black, fill=white] (D1) [below=of C1] {$s_3$};
\node[rectangle, draw=black, fill=white]  (A1) at (0,0) {$s_0$};
%\node[rectangle, draw=black, fill=white] (B1) [below=of A1] {$s_1$};
\node[rectangle, draw=black, fill=white] (B1) at (0,-2.5) {$s_1$};
%\node[rectangle, draw=black, fill=white] (C1) [right=of A1] {$s_2$};
\node[rectangle, draw=black, fill=white]  (C1) at (2.5,0)  {$s_2$};
%\node[rectangle, draw=black, fill=white] (D1) [below=of C1] {$s_3$};
\node[rectangle, draw=black, fill=white] (D1) at (2.5,-2.5) {$s_3$};

\path[->,thick] (A1) edge[bend left=10] node {3} (B1); 
\path[->,thick] (A1) edge[bend left=10] node {2} (C1);
\path[->,dotted] (B1) edge[bend left=10] node {1} (A1);
\path[->,thick] (C1) edge[bend left=10] node {1} (A1);
\path[->,thick] (C1) edge[bend left=10] node {1} (D1);
%\path[->,dotted] (B1) edge node {2} (D1); 
\path[->,thick] (A1) edge node {2} (D1); 
\path[->,dotted] (D1) edge[bend left=10] node {1} (C1); 

\path[->,thick] (D1) edge[loop below] node[align=left,pos=.5] {2} (D1);
\path[->,thick] (B1) edge[loop below] node[align=left,pos=.5] {3} (B1);
\path[->,thick] (A1) edge[loop above] node[align=left,pos=.5] {3} (A1);
\path[->,thick] (C1) edge[loop above] node[align=left,pos=.5] {2} (C1);

%\node at ($(B1)!0.5!(D1)-(0,1.5)$) {(a)};

% B

%\node[rectangle, draw=black, fill=white] (A2) [right=6cm of A1] {$s_0$};
%\node[rectangle, draw=black, fill=white] (B2) [below=of A2] {$s_1$};
%\node[rectangle, draw=black, fill=white] (C2) [right=of A2] {$s_2$};
%\node[rectangle, draw=black, fill=white] (D2) [below=of C2] {$s_3$};

\node[rectangle, draw=black, fill=white] (A2) at (5,0) {$s_0$};
%\node[rectangle, draw=black, fill=white] (B2) [below=of A2] {$s_1$};
 
\node[rectangle, draw=black, fill=white] (B2) at (5,-2.5) {$s_1$};
%\node[rectangle, draw=black, fill=white] (C2) [right=of A2] {$s_2$};
\node[rectangle, draw=black, fill=white] (C2) at (7.5,0) {$s_2$};
%\node[rectangle, draw=black, fill=white] (D2) [below=of C2] {$s_3$};
\node[rectangle, draw=black, fill=white] (D2) at (7.5,-2.5) {$s_3$};

\path[->,thick] (A2) edge[bend left=10] node {3} (B2); 
\path[->,thick] (A2) edge[bend left=10] node {2} (C2);
\path[->,dotted] (B2) edge[bend left=10] node {1} (A2);
\path[->,thick] (C2) edge[bend left=10] node {1} (A2);
\path[->,dotted] (C2) edge[bend left=10] node {1} (D2);
%\path[->,dotted] (B2) edge node {2} (D2); 
\path[->,dotted] (A2) edge node {2} (D2); 
\path[->,dotted] (D2) edge[bend left=10] node {1} (C2); 

\path[->,thick] (D2) edge[loop below] node[align=left,pos=.5] {2} (D2);
\path[->,thick] (B2) edge[loop below] node[align=left,pos=.5] {3} (B2);
\path[->,thick] (A2) edge[loop above] node[align=left,pos=.5] {3} (A2);
\path[->,thick] (C2) edge[loop above] node[align=left,pos=.5] {2} (C2);

%\node at ($(B2)!0.5!(D2)-(0,1.5)$) {(b)};

%(n,m) cost n to remove, m to add arrow
%Omitted arrows/costs are too costly for them to afford to change
%Assume budget = 9 each
\end{tikzpicture}
%}
\caption{Models $\M_{1}$ (left) and $\M_{2}$ (right). %representing a computational system. 
Arrows with solid lines represent live transitions in the system, whereas dotted lines depict possible new transitions within cost 3. %\munyquem{ Costs for adding and removing errors are the same. ?! Yes, costs between states are the same for both angel and demon. We can mention that we can generalise it and the results will go through. MM: yes, let us add a footnote (or remark). I'll add it after the definition of a model.   }
}
\label{fig:example}
\end{figure} 
Assume that the user is trusted to access the system, but not %authorised to 
the admin module. In this case, the security engineer has a strategy with cost 2 that removes the transitions from $s_2$ to $s_3$ and $s_0$ to $s_3$ and prevents global access to the admin state from other states. The order of removing transitions depends on the current state. The resulting model $\M_2$ after two steps of the game is shown in Figure \ref{fig:example}. We can thus see that  $(\M_1,s) \models \estrat 2 \mathsf{G} \,\neg \, admin$, for any $s \neq s_3$. 

%Notice that $(\M_1,s_3) \not \models \estrat 2 \mathsf{G} \,\neg\, admin$. Although the transition loop in $s_3$ has a small cost, it cannot be removed under demonic strategies, as doing so would break the seriality of the model.

% i.e., the security has a strategy to prevent the access to the admin state globally. 
%, that is, there is a strategy for the security engineer to ensure that  

%Meaning of labels: q0 l = authentication - login state; q1 e = error state; server s= the user has access to the system; q3 a = the user has admin access. This represents an access to a system. Arrows are labeled with the intuition of what the traveler (the user) would be doing.  Costs to be added. In blue, an arrow that could be added to allow the traveler to try again.

\end{example}
%We define the operator $\astrat n  \psi$ as $\neg\estrat{n} \neg \psi$ where $\neg \psi$ is 
%\begin{align*}
 %   \neg( \nextt \varphi) &= \nextt \neg \varphi 
  %      \\
   %     \neg (\varphi_1 \until \varphi_2) &= \neg \varphi_1 \release \neg \varphi_2 
    %    \\
     %   \neg (\varphi_1 \release \varphi_2) &= \neg \varphi_1 \until \neg \varphi_2
%\end{align*}

%\rust{Observe that we can define the standard modal $\square$ and $\Diamond$ modalities using demonic strategies with 0 resources. Indeed, take $\square \varphi := \estrat 0 \nextt \varphi$ and $\Diamond \varphi:= \astrat 0 \nextt \varphi$. It is easy to verify that the box and diamond have exactly the intended semantics as the only strategy the demon can play is keeping the model intact.} %\munyque{Do we mention the shortcuts (eventually, implication, ..)?} Yes, need to do it when define languages

\subsection{Strategic Construction Logic}
%\munyque{Changed "Rebuild Strategic Logic£ to "Strategic Rebuild Logic" to match the name of the first logic.}

\def\angstrat{\mathcal{S}}

Now we introduce Strategic Construction Logic that, in some sense, is dual to SDL. Whilst in SDL the demon can remove edges, in Strategic Construction Logic \textit{the angel} can strategically add edges.

\begin{definition}[Angelic Strategy]
Given a countably infinite sequence $\pi$ of pointed models, $\pi$ is an \emph{increasing model path with cost $n$} iff for every $i\geqslant 0$ we have that $\pi(i) \racc{n} \pi(i+1)$. As in Definition \ref{def:demstrat}, we write $\pi^\M$ to denote the corresponding sequence of models.

    %An \emph{angelic strategy} is a function $\angstrat$ that given a history $h=h(1),\ldots, h(n)$ outputs a set of pairs of states of $h^\M (n)$ that do not belong to $\xrightarrow{h^\M (n)}$. If $\pi$ is an increasing model path, then $\pi$ is compatible with $\angstrat$ if and only if for every $i\geqslant |\pi|$, $\angstrat(\pi(\leq i))= A$ implies $\pi^\M(i+1) = \pi^\M (i) \cup A$. An angelic strategy will be called an \emph{n-strategy} iff every path that is compatible with the strategy has cost at most $n$. We let $Out(\angstrat,\pi)$ denote the set of increasing paths starting at $\pi$ that are compatible with $\angstrat$.

    An \emph{angelic strategy} is a function $\angstrat$ that, given a pointed model $(\M,s)$, returns a set of edges $A\subseteq ((S^\M \times S^\M)\setminus \to^\M) $. 
    
    An increasing model path $\pi$ is compatible with $\angstrat$ iff for all $i \in \mathbb{N}$, we have that  $\angstrat(\pi(i))= A$ implies $\pi^\M(i+1) = \pi^\M (i) \cup A$. An angelic strategy $\angstrat$ will be called an \emph{n-strategy} iff every increasing model path that is compatible with the strategy has cost  $n$. We let $Out(\angstrat,(\M,s))$ denote the set of increasing model paths starting at $(\M,s)$ that are compatible with $\angstrat$.
 \end{definition}

\begin{definition}[Strategic Construction Logic]
    Formulae of \textit{Strategic Construction Logic} (SCL) are defined similarly to those of Strategic Deconstruction Logic in Definition \ref{def:sdl},
with the difference that state formulae of the form $\estrat{n}\psi$ are replaced by state formulae for strategic angelic operators $\angel{n}\psi$.  Constructs $\angel{n}\psi$ are read as `there is an angelic $n$-strategy such that for all moves of the traveller, $\psi$ holds'. The dual of the operator is defined as $\allangel n \psi := \lnot \angel n \lnot \psi$.
\end{definition}

%Given a model $\M$, and $A\subseteq S^\M \times S^\M$ we write $\M \cup A$ as a shorthand for $(S, \to \! \cup A, \mathcal{V}, \mathcal{C})$. If $(\M,s)$ and $(\M',s')$ are two pointed models we say that $(\M',s')$ is \textit{rebuild accessible} from $(\M,s)  $ with cost $n$, denoted $(\M,s) \racc n (\M',s')$, if and only if $\M'=\M \cup A$ for some $A\subseteq ((S^\M \times S^\M)\setminus \to^\M) $, $\C(A) \leqslant n$, and $s\xrightarrow{\M'} s$. We will also call $\M'$ an \textit{$n$-supermodel} of $\M$. While the agent moving after the demon cuts edges was called an attacker, we call the same agent in the context of SCL a \emph{traveller}.

%Intuitively, the $n$-supermodel $\M'$ is obtained from $\M$ by adding a new set of edges whose total cost does not exceed $n$, and $s'$ is a successor of $s$ with respect to the resulting transition relation $\xrightarrow{\M'}$. 
%Observe that any $n$-supermodel $\M'$ of $\M$ is serial.

 The semantics of SCL is defined similarly to the semantics of SDL (Definition \ref{def::sdl_sem}), except for the case of the strategic operator. %that we provide below. 
%\munyque{$>$ The semantics of SCL is defined analogously to the semantics of SDL (Definition \ref{def::sdl_sem}), except by the strategic  operator, which is defined as follows:  }
 \begin{definition}[SCL Semantics]
    Let $(\M,s)$ be a pointed model. 

$\begin{array}{l l l}
   (\M,s) \models \angel n \psi &  \text{iff} & \text{there is an $n$-strategy } \angstrat \text{ s.t. for all}
   \\
   &  & \pi\in Out(\angstrat, (\M,s)) \text{ we have } \pi\models \psi
\end{array}$

Similarly to SDL, we can define the standard modal box and diamond in SCL as $\square \varphi := \angel 0 \nextt \varphi$ and $\Diamond \varphi:= \allangel 0 \nextt \varphi$.

\end{definition}

\begin{example}{(Access control, cont.)}
%Let us resume Example \ref{sec:example1} and consider the model $\M_1$. Once the user enters the failure state (e.g., because she typed the wrong password), she remains permanently blocked there.  Assume we want to allow the user to exit the failure state in one step. 
Let us consider model $\M_1$ from Example \ref{sec:example1} again. Once the user enters the failure state $s_1$ by, e.g., typing a wrong password, she remains permanently blocked there.  Assume now that we want to allow the user to exit the failure state in one step.
Adding a transition from $s_1$ to $s_0$, which has a cost of 1, would allow the user to retry entering the server module. The resulting model $\M_3$ after applying this angelic strategy is shown in Figure \ref{fig:example2}. 
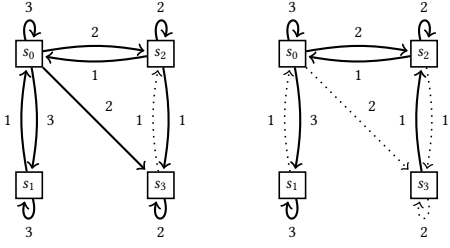
\begin{figure}[h!]
\centering
%\scalebox{0.8}{
 \begin{tikzpicture}[->,shorten >=1pt,auto,node distance=3cm and 3cm, semithick,minimum size=0.5cm, scale=0.7, transform shape]%,minimum size=1.cm

\node[rectangle, draw=black, fill=white]  (A1) at (0,0) {$s_0$};
%\node[rectangle, draw=black, fill=white] (B1) [below=of A1] {$s_1$};
\node[rectangle, draw=black, fill=white] (B1) at (0,-2.5) {$s_1$};
%\node[rectangle, draw=black, fill=white] (C1) [right=of A1] {$s_2$};
\node[rectangle, draw=black, fill=white]  (C1) at (2.5,0)  {$s_2$};
%\node[rectangle, draw=black, fill=white] (D1) [below=of C1] {$s_3$};
\node[rectangle, draw=black, fill=white] (D1) at (2.5,-2.5) {$s_3$};

\path[->,thick] (A1) edge[bend left=10] node {3} (B1); 
\path[->,thick] (A1) edge[bend left=10] node {2} (C1);
\path[->,thick] (B1) edge[bend left=10] node {1} (A1);
\path[->,thick] (C1) edge[bend left=10] node {1} (A1);
\path[->,thick] (C1) edge[bend left=10] node {1} (D1);
%\path[->,dotted] (B1) edge node {2} (D1); 
\path[->,thick] (A1) edge node {2} (D1); 
\path[->,dotted] (D1) edge[bend left=10] node {1} (C1); 

\path[->,thick] (D1) edge[loop below] node[align=left,pos=.5] {2} (D1);
\path[->,thick] (B1) edge[loop below] node[align=left,pos=.5] {3} (B1);
\path[->,thick] (A1) edge[loop above] node[align=left,pos=.5] {3} (A1);
\path[->,thick] (C1) edge[loop above] node[align=left,pos=.5] {2} (C1);

%\node at ($(B1)!0.5!(D1)-(0,1.5)$) {(a)};

% B

%\node[rectangle, draw=black, fill=white] (A2) [right=6cm of A1] {$s_0$};

\node[rectangle, draw=black, fill=white] (A2) at (5,0) {$s_0$};
%\node[rectangle, draw=black, fill=white] (B2) [below=of A2] {$s_1$};
 
\node[rectangle, draw=black, fill=white] (B2) at (5,-2.5) {$s_1$};
%\node[rectangle, draw=black, fill=white] (C2) [right=of A2] {$s_2$};
\node[rectangle, draw=black, fill=white] (C2) at (7.5,0) {$s_2$};
%\node[rectangle, draw=black, fill=white] (D2) [below=of C2] {$s_3$};
\node[rectangle, draw=black, fill=white] (D2) at (7.5,-2.5) {$s_3$};

\path[->,thick] (A2) edge[bend left=10] node {3} (B2); 
\path[->,thick] (A2) edge[bend left=10] node {2} (C2);
\path[->,dotted] (B2) edge[bend left=10] node {1} (A2);
\path[->,thick] (C2) edge[bend left=10] node {1} (A2);
\path[->,dotted] (C2) edge[bend left=10] node {1} (D2);
%\path[->,dotted] (B2) edge node {2} (D2); 
\path[->,dotted] (A2) edge node {2} (D2); 
\path[->,thick] (D2) edge[bend left=10] node {1} (C2); 

\path[->,dotted] (D2) edge[loop below] node[align=left,pos=.5] {2} (D2);
\path[->,thick] (B2) edge[loop below] node[align=left,pos=.5] {3} (B2);
\path[->,thick] (A2) edge[loop above] node[align=left,pos=.5] {3} (A2);
\path[->,thick] (C2) edge[loop above] node[align=left,pos=.5] {2} (C2);

\node at ($(B2)!0.5!(D2)-(0,1.5)$) {};

\end{tikzpicture}
%}
\caption{Model $\M_3$ (left) and  $\M_4$ (right),
obtained from changing $\M_1$ and $\M_2$ (Figure \ref{fig:example}), respectively.}
\label{fig:example2}
\end{figure}
Notice that the angel does not ensure the traveller will move out of the error state, i.e., 
 $(\M_3,s_1) \not \models \angel  1 \mathsf{X} \,\neg\, error$. That is because angelic strategies can create new possibilities for the traveller, but they cannot force her to take them.  However, we can easily check  that %$(\M_3,s_1)  \models \angel  1 \mathsf{X} \Diamond \,\neg\, error$
  $\angel  1 \mathsf{G} (error \to \Diamond \,\neg\, error)$ holds in all states of model $\M_1$, that is, the angel has a strategy to ensure that every time the traveller enters a failure state, she can leave it in one step. Indeed, it is enough for the angel to restore the transition from $s_1$ and $s_0$, resulting in model $\M_3$, to satisfy $\mathsf{G} (error \to \Diamond \,\neg\, error)$. %in the next state, it is possible for the traveller to not be in the failure state.

 %Agents with angelic strategies can also collaborate with a malicious attacker to enable access to unauthorized states, creating vulnerabilities to the system. For instance,  producing a  transition from $s_1$ to $s_3$ would let the user to go to the admin module without passing through the authentication stage.

%"There is an angel strategy s.t. X diamond $(l \lor s)$" (the user is allowed to retry..)
\end{example}

 \subsection{Strategic Update Logic}
Having defined  separate logics for the demon and the angel, we now combine the two, allowing them to cooperate or compete with one another. In this, we take inspiration from logics for MAS (like alternating-time temporal logic (ATL) \cite{alur2002}, coalition logic (CL) \cite{pauly02}, and strategy logic (SL) \cite{mogavero10}), and assume that the angel and the demon execute their actions \textit{concurrently}. Our new modalities are inspired by those of ATL and CL, $\langle \! [ C ] \! \rangle \varphi$, that mean `there is a joint strategy of agents in coalition $C$ such that no matter what agents outside of the coalition do, $\varphi$ holds'. In our case, agents with strategies are the angel and the demon.  

\begin{definition}[Update Model Paths]
    For a countably infinite sequence $\pi$ of pointed models, we say that $\pi$ is an \textit{update model path with costs $n$ and $m$} iff  for all $i \in \mathbb{N}$ we have that $\pi(i) \overset{n,m}{\Rightarrow} \pi(i+1)$. 
Given an angelic $n$-strategy $\angstrat$ and a demonic $m$-strategy $\strat$, an update model path $\pi$ with costs $n$ and $m$ is compatible with the strategies iff  for all $i \in \mathbb{N}$, we have $\pi^\M (i+1) = (\pi^\M (i) \setminus \strat (\pi(i))) \cup \angstrat (\pi(i))$.
Finally, $Out(\angstrat, \strat, (\M,s))$ is the set of update model paths starting at $(\M,s)$ and compatible with strategies $\angstrat$ and $\strat$.
\end{definition}

\begin{definition}[Strategic Update Logic]
    \emph{Strategic Update Logic} (SUL) is defined similarly to SDL and SCL, with the difference that modalities for strategic demonic or angelic operators are substituted with formulas $\esetstrat n m \psi$, where $n,m \geqslant 0$, and $C \subseteq \{\pentacle, \angelsymbol\}$. In particular, we can have four variations of the formula: 1) $C = \{\pentacle\}$, and we will write $\demonstrat n m \psi$ meaning that `there is a demonic $n$-strategy such that for all angelic $m$-strategies and all moves of the traveller, $\psi$ holds'; 2) $C = \{\angelsymbol\}$, and we will write $\angelstrat n m \psi$ with the meaning as in the previous case with demonic and angelic strategies swapped; 3) $C = \{\pentacle, \angelsymbol\}$, which we will write as ${\angeldemonstrat n m \psi}$ meaning `there is an angelic $n$-strategy and a demonic $m$-strategy such that for all moves of the traveller, $\psi$ holds'; and finally 4) $C = \emptyset$, denoted as $\emptystrat n m \psi$ with the meaning as in the previous case with existential quantifiers swapped for universal ones. 
    As usual, we will denote the dual of the combined strategic operator as $\asetstrat n m \psi := \lnot \esetstrat n m \lnot \psi$.
\end{definition}

%Given a model $\M$, its $n$-supermodel $\M_1$ obtained by adding the set of edges $A$ and $m$-submodel $\M_2$ obtained by removing the set of edges $B$, we will denote by $\M_1 \star \M_2 = (\M \setminus B) \cup A$ the resulting model after the angel and the demon has added or removed their sets of edges. We will call  $\M_1 \star \M_2$ an \textit{$n$-$m$-update}. Note that $\M_1 \star \M_2$ is well-defined as the sets of edges that the angel and the demon manipulate are disjoint.

%For two pointed models $(\M, s)$ and $(\M', s')$, we say that $(\M', s')$ is \textit{update accessible} from $(\M, s)$ with costs $n$ and $m$, denoted byf $(\M, s) \overset{n,m}{\Rightarrow}(\M', s')$, if and only if $\M'$ is an $n$-$m$-update of $\M$, and $s \xrightarrow{\M'} s'$.

 \begin{definition}[SUL Semantics]
    Let $(\M,s)$ be a pointed model. %We present the semantics for all four cases of the strategic update modality.
    We present only cases of $\angelstrat n m \psi$ and $\angeldemonstrat n m \psi$, and the semantics for $\demonstrat n m$ and $\emptystrat n m \psi$ are defined analogously by swapping demonic and angelic strategies, as well as quantifiers.
    
$\begin{array}{l l l}
   (\M,s) \models \angelstrat n m \psi &  \text{iff} & \text{there is an $n$-strategy }  \angstrat \text{ s.t. }\\ 
   &&\text{for all $m$-strategies } \strat \text{ and} \\
   && \text{for all } \pi \in Out(\angstrat, \strat, (\M,s))\\
   &  &  \text{we have that } \pi\models \psi\\
      %(\M,s) \models \demonstrat n m \psi &  \text{iff} & \text{there is an $n$-strategy }  \strat \text{ s.t. }\\ 
   %&&\text{for all $m$-strategies } \angstrat \text{ and} \\
   %&& \text{for all } \pi \in Out(\angstrat, \strat, (\M,s))\\
   %&  &  \text{we have that } \pi\models \psi\\
      (\M,s) \models \angeldemonstrat n m \psi &  \text{iff} & \text{there is an $n$-strategy }  \angstrat \text{ and }\\ 
   &&\text{there is an $m$-strategy } \strat \text{ s.t.} \\
   && \text{for all } \pi \in Out(\angstrat, \strat, (\M,s))\\
   &  &  \text{we have that } \pi\models \psi\\
    %  (\M,s) \models \emptystrat n m \psi &  \text{iff} & \text{for all $n$-strategies }  \angstrat \text{ and }\\ 
   %&&\text{for all $m$-strategies } \strat \text{ and} \\
   %&& \text{for all } \pi \in Out(\angstrat, \strat, (\M,s))\\
   %&  &  \text{we have that } \pi\models \psi\\
\end{array}$

\end{definition}

\begin{example}{(Access control, cont.)}
 Cooperation between the angel and demon allows for a dynamic, synchronised access control to the system. 
 Let us consider model $\M_2$ in Figure \ref{fig:example}. The security engineer, the demon, on her own can not make the user leave the admin module, i.e. $(\M_2,s_3) \not\models \demonstrat
 2 2 \mathsf{F} \,\neg \, admin$ as the angel can play the strategy of not modifying the model.  However, this goal can be achieved when both the demon and the angel cooperate. Particularly, in the first step, the angel creates transition $s_3 \to s_2$, and the demon chooses a $0$-submodel (i.e. does nothing), and in the next step of the game the demon removes the self-loop at $s_3$. This strategy results in model $\M_4$ in Figure \ref{fig:example2}. 
 Thus, we have that  
   $(\M_2,s_3) \models \angeldemonstrat
 2 2 \mathsf{F} \,\neg \, admin$, as the only thing the user can do once the model is updated to $(\M_4,s_3)$ is to move to state $s_2$. 

 %e.g., granting temporary admin access to a user. 
 %There is a joint demon and angel strategy st $\neg e \until a$". They can prevent the user from entering an error state and give him admin access.
 With SUL we can also capture adversarial interactions. For instance, the angel can collaborate with a malicious attacker in $(\M_2,s_1)$ to enable access to unauthorized states, creating vulnerabilities in the system. For example, $(\M_2,s_1) \not \models \demonstrat 2 2 \mathsf{G} \,\neg \, admin$. Indeed, to prevent the attacker aided by the angel from reaching state $s_3$ from $s_1$, the demon can remove the transition $s_0 \to s_2$, while at the same time the angel builds a bridge $s_1\to s_0$. Now, if the attacker is in the state $s_0$, the angel can, for example, add a bridge $s_0 \to s_3$ and no matter what the demon does at the same time, the attacker will get access to the admin state. 
 However, the security engineer can prevent this attack if the angel has fewer resources. We can easily check that  $(\M_2,s_1) \models \demonstrat 2 1 \mathsf{G} \,\neg \, admin$. Indeed, the demon can remove the transition from $s_0$ to $s_2$, and no matter what the angel does with 1 resource, she cannot restore a path to state $s_3$.

%When angel and demon cooperate, they are managing the user access. We can also consider the case in which they are adversarial and the angel is a hacker trying to give his friend admin access.

\end{example}

 \section{Expressivity}
 \label{sec:exp}

 In this section, we compare our new logics to each other as well as to established logics in the literature. 

\begin{definition}[Expressivity]
Let $\mathcal{L}_1$ and $\mathcal{L}_2$ be two languages, and let $\varphi \in \mathcal{L}_1$ and $\psi \in \mathcal{L}_2$. 
We say that $\varphi$ and $\psi$ are \emph{equivalent}, when for all models $(\M,s)$: $(\M,s) \models \varphi$ if and only if $(\M,s) \models \psi$.

If for every $\varphi \in \mathcal{L}_1$ there is an equivalent $\psi \in \mathcal{L}_2$, we write $\mathcal{L}_1 \preccurlyeq \mathcal{L}_2$ and say that $\mathcal{L}_2$ is \emph{at least as expressive as} $\mathcal{L}_1$. We write $\mathcal{L}_1 \prec \mathcal{L}_2$ iff $\mathcal{L}_1 \preccurlyeq \mathcal{L}_2$ and $\mathcal{L}_2 \not \preccurlyeq \mathcal{L}_1$, and we say that $\mathcal{L}_2$ is \emph{strictly more expressive than} $\mathcal{L}_1$. Finally, if $\mathcal{L}_1 \not \preccurlyeq \mathcal{L}_2$ and $\mathcal{L}_2 \not \preccurlyeq \mathcal{L}_1$, we say that $\mathcal{L}_1$ and $\mathcal{L}_2$ are \emph{incomparable} and write $\mathcal{L}_1 \not \approx \mathcal{L}_2$.
\end{definition}

\textbf{{Computation Tree Logic.} }We start by showing that all of SDL, SCL, and SUL are strictly more expressive than the classic \textit{Computation Tree Logic} (CTL) \cite{ctl}. For the proof that our logics are at least as expressive s CTL, we argue, similarly to the argument for OL \cite{CattaLM23}, that CTL is a fragment of all of SDL, SCL, and SUL. To show that there are properties that our logics can express while CTL cannot, we use the model changes.

\begin{theorem}
\label{thm:ctlexp}
    CTL $\prec$ SDL, CTL $\prec$ SCL, and CTL $\prec$ SUL. 
\end{theorem} 

\begin{proof}
    To see that CTL is a fragment of our logics it is enough to define a truth-preserving translation function $t$ from CTL to our logics. We omit Boolean cases as they are immediate. Now, to deal with CTL path quantifiers, we can employ demonic and angelic strategies over 0 resources. In particular, $t(\mathsf{A X} \varphi) = \heartsuit \nextt t(\varphi)$ and $t(\mathsf{A} (\varphi \until \psi)) = \heartsuit  (t(\varphi) \until t(\psi))$, where $\heartsuit \in \{\estrat 0, \angel 0, {\angeldemonstrat 0 0}\}$ depending on the logic in question. It is immediate, by the definition of the semantics, that such a recursive translation is truth-preserving.

    Now, for each of SDL, SCL, and SUL, we show that there are models that they can distinguish and CTL cannot. First, we start with SDL. Consider an SDL formula $\estrat 1 \nextt p$, and %%consider towards a contradiction 
    assume towards a contradiction that there is an equivalent formula $\varphi$ of CTL. Then, consider models $\M_1$ and $\M_2$ depicted in Figure \ref{fig:CTLvsSDL}.

 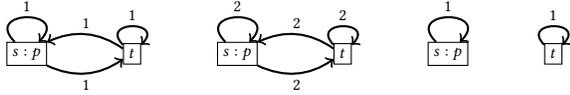
\begin{figure}[h!]
\centering
%\scalebox{0.8}{
   \begin{tikzpicture}[scale=0.7, transform shape]
\node[rectangle,draw=black](s) at (0,0) {$s:p$};
\node[rectangle,draw=black](t) at (2,0) {$t$};

\draw[->,thick,bend right] (s) to node[below] {$1$} (t);
\draw[->,thick,bend right] (t) to node[above] {$1$} (s);
\draw [<-,thick](s) to [loop above, out=45, in=135, looseness = 5] node[above] {$1$} (s); %edited loop above right
\draw [<-,thick](t) to [loop above, out=45, in=135, looseness = 5] node[above] {$1$} (t);

\node[rectangle,draw=black](s) at (4,0) {$s:p$};
\node[rectangle,draw=black](t) at (6,0) {$t$};

\draw[->,thick,bend right] (s) to node[below] {$2$}  (t);
\draw[->,thick,bend right] (t) to node[above] {$2$} (s);
\draw [<-,thick](s) to [loop above, out=45, in=135, looseness = 5] node[above] {$2$} (s); %edited loop above right
\draw [<-,thick](t) to [loop above, out=45, in=135, looseness = 5] node[above] {$2$} (t);

\node[rectangle,draw=black](s) at (8,0) {$s:p$};
\node[rectangle,draw=black](t) at (10,0) {$t$};

\draw [<-,thick](s) to [loop above, out=45, in=135, looseness = 5] node[above] {$1$} (s); %edited loop above right
\draw [<-,thick](t) to [loop above, out=45, in=135, looseness = 5] node[above] {$1$} (t);
\end{tikzpicture}
 
%}

\caption{Models $\M_1$ (left), $\M_2$ (middle), and $\M_3$ (right). Atom $p$ is true in states $s$. Cost of all possible edges in $\M_3$ is 1.}
\label{fig:CTLvsSDL}
\end{figure} 

The two models are isomorphic with the only difference that the cost of every edge in $\M_1$ is 1, and the cost of every edge in $\M_2$ is 2. Since in CTL we do not have access to the cost of edges, $(\M_1,s) \models \varphi$ if and only if $(\M_2,s) \models \varphi$.

It is also easy to verify that $(\M_1,s)  \models \estrat 1 \nextt  p$. Indeed, the demon can remove the edge $s \to t$, and thus the traveller will never reach state $t$, where $p$ is false. Also, it is immediate that $(\M_2,s) \not \models \estrat 1 \nextt \lnot p$ because the cost of every edge in the model is 2 and there is nothing the demon can do to modify the model. 

We now turn to SCL and SUL, and consider the following two formulas: ${\angel 1 \nextt \allangel 0 \nextt  \lnot p}$ of SCL and ${\angelstrat 1 0 \nextt  [ \! \langle \angelsymbol^0, \pentacle^0 \rangle \! ] \nextt \lnot p}$ of SUL. These formulas mean that the angel can add an arrow such that in the new updated model, the traveller can reach a $\lnot p$-state.  Assume towards a contradiction that there is some equivalent formula $\varphi$ of CTL for each of them. %Let us consider 
Take model $\M_3$ in Figure \ref{fig:CTLvsSDL} and model $\M_4$, which is like $\M_3$ but contains the single state $s$. 
\iffalse
 \begin{figure}[h!]
\centering
%\scalebox{0.8}{
   \begin{tikzpicture}[scale=0.7, transform shape]
\node[rectangle,draw=black](s) at (0,0) {$s:p$};
\node[rectangle,draw=black](t) at (2,0) {$t$};

\draw [<-,thick](s) to [loop above, out=45, in=135, looseness = 5]  (s); %edited loop above right
\draw [<-,thick](t) to [loop above, out=45, in=135, looseness = 5]  (t);
\end{tikzpicture}
\hspace{3cm}
   \begin{tikzpicture}[scale=0.7, transform shape]
\node[rectangle,draw=black](s) at (0,0) {$s:p$};

\draw [<-,thick](s) to [loop above, out=45, in=135, looseness = 5]  (s); %edited loop above right
\end{tikzpicture}
 
%}

\caption{Models $\M_3$ (left, with states $s$ and $t$) and $\M_4$ (right, with only state $s$). Cost of every edge is 1.}
\label{fig:CTLvsSCL}
\end{figure} 
\fi

Model $\M_4$ is just a single state $s$ with a reflexive arrow. Model $\M_3$ is the state $s$ as well as the state $t$, both having only reflexive edges. It is immediate that $(\M_3,s) \models \varphi$ if and only if $(\M_4,s) \models \varphi$.

At the same time, we have that $(\M_4,s) \not \models {\angel 1 \nextt \allangel 0 \nextt  \lnot p}$ and $(\M_4,s) \not \models {\angelstrat 1 0 \nextt  [ \! \langle \angelsymbol^0, \pentacle^0 \rangle \! ] \nextt \lnot p}$ as there is simply no state in model $\M_4$ satisfying $\lnot p$. 
On the other hand, $(\M_3,s) \models {\angel 1 \nextt \allangel 0 \nextt  \lnot p}$, as the angel can add the edge $s \to t$ to make the $\lnot p$-state $t$ accessible for the traveller. In particular, after the angel adds the edge $s \to t$, we check all next-time paths the traveller can take. She can either stay in $s$ or move to $t$. In both cases, ${\allangel 0 \nextt  \lnot p}$ is satisfied, as with 0 resources the angel cannot modify the model anymore, and there a next-time path for the traveller to reach the $\lnot p$-state $t$. A similar reasoning can be used to see that $(\M_3,s) \models {\angelstrat 1 0 \nextt  [ \! \langle \angelsymbol^0, \pentacle^0 \rangle \! ] \nextt \lnot p}$.
\end{proof}

\textbf{Angels and Demons.} We show that once pitched against one another, the new logics of angels and demons are indeed different. In particular, we argue that SDL and SCL are incomparable and that SUL is strictly more expressive than both of them.

\begin{theorem}
\label{thm:SCLvsSDL}
    SDL $\not \approx$ SCL, SDL $\prec$ SUL, and SCL $\prec$ SUL.
\end{theorem}

\begin{proof}
    We can reuse the arguments from the proof of Theorem \ref{thm:ctlexp}. To see that SDL $\not \preccurlyeq$ SCL, we recall that models $\M_1$ and $\M_2$ are distinguishable by an SDL formula. 
    That no SCL formula can distinguish the two models follows from the fact that the models differ only in the costs of their edges, and since the relations for both models are universal, the angel cannot modify the model.

    To show that SCL $\not \preccurlyeq$ SDL, we recall models $\M_3$ and $\M_4$ from the proof of Theorem \ref{thm:ctlexp} that are distinguishable by an SCL formula. That no SDL formula can distinguish $(\M_3, s)$ and $(\M_4, s)$ is immediate by the semantics of the demonic operator (the demon cannot remove any further edges). Hence, SDL $\not \approx$ SCL.

    Finally, observe that both SDL and SCL are fragments of SUL via a translation $t(\estrat n \varphi) = \angeldemonstrat 0 n t(\varphi)$ and $t(\angel n \varphi) = {\angeldemonstrat n 0} t(\varphi)$, and therefore SDL $ \preccurlyeq$ SUL and SCL $ \preccurlyeq$ SUL.  Hence, we also have that  SCL $\not \preccurlyeq$ SDL implies SUL $\not \preccurlyeq$ SDL, and SDL $\not \preccurlyeq$ SCL implies SUL $\not \preccurlyeq$ SCL. Putting everything together, we conclude that SDL $\prec$ SUL, and SCL $\prec$ SUL.
\end{proof}

\textbf{Relation to Obstruction Logic.}
The semantics of the strategic operators $\langle \dagger^n\rangle \psi$ and $[\dagger^n ]\psi$ of \textit{Obstruction Logic} (OL) \cite{CattaLM23} is similar to the semantics of the SDL modalities $\estrat n \psi$ and $\astrat n \psi$ with the crucial difference that after the demon disables some edges and the traveller makes a move, the edges in OL are restored. Hence, in OL, the changes in the given model are \textit{not permanent}. However, in the proof of SDL $\not \preccurlyeq$ SCL in Theorem \ref{thm:SCLvsSDL}, we used only next-time temporal modalities in the SDL formula. Hence, we can use the same argument and the same pair of models, $\M_1$ and $\M_2$, to show that OL $\not \preccurlyeq$ SCL. For the other direction, i.e., SCL $\not \preccurlyeq$ OL, the argument is similar to the one for SCL $\not \preccurlyeq$ SDL that uses models $\M_3$ and $\M_4$. The same argument implies that SUL $\not \preccurlyeq$ OL.

\begin{theorem}
\label{OLvsSCL}
    OL $\not \approx$ SCL  and SUL $\not \preccurlyeq$ OL.  
\end{theorem}

The relation between OL and SDL is a more nuanced one. In the proof of SDL $\not \preccurlyeq$ OL, we utilise the fact that the model change in SDL is permanent and hence we can `remember' the removed edges, which is impossible in OL, where edges are restored before each new action of the demon.

\begin{restatable}{theorem}{sdlvsol}
\label{thm:SDLvsOL}
    SDL $\not \preccurlyeq$ OL.
\end{restatable}

\begin{proof}
    We provide the idea behind the proof, and the full proof can be found in Technical Appendix. Consider an SDL formula $\estrat 1 \mathsf{F} p$, and assume towards a contradiction that there is an equivalent formula $\varphi$ of OL. Since formulas of OL are finite, and each OL modality has a finite resource bound, we can assume that $n$ is the greatest $n \in \mathbb{N}$ appearing in $\varphi$. Now consider models $\M_{n+2}$ and $\M_{n+3}$ in Figure \ref{fig:pspaceex3}.
    \begin{figure}[h!]
\centering
%\scalebox{0.8}{
   \begin{tikzpicture}[scale=0.7, transform shape]
%\node[circle,draw=black, minimum size=4pt,inner sep=0pt, fill = black, label=above:{$s_1$}](s) at (0,0) {};
\node[rectangle,draw=black](s) at (0,0) {$s_1$};
%\node[circle,draw=black, minimum size=4pt,inner sep=0pt,, fill = black , label=left:{$s_2$}](a10) at (0,-1) {};
\node[rectangle,draw=black](a10) at (0,-1) {$s_2$};
\node(dots) at (0,-2) {$...$};
%\node[circle,draw=black, minimum size=4pt,inner sep=0pt,, fill = black , label=left:{$s_n$}](sn0) at (0,-3) {};
\node[rectangle,draw=black](sn0) at (0,-3) {$s_n$};
%\node[circle,draw=black, minimum size=4pt,inner sep=0pt,, fill = black , label=left:{$s_{n+1}$}](sn) at (0,-4) {};
\node[rectangle,draw=black](sn) at (0,-4) {$s_{n+1}$};
%\node[circle,draw=black, minimum size=4pt,inner sep=0pt, fill = black, label=below:{$t_1: p$}](t1) at (-2,-5) {};
\node[rectangle,draw=black](t1) at (-2,-5) {$t_1: p$};
%\node[circle,draw=black, minimum size=4pt,inner sep=0pt, fill = black, label=below:{$t_2$}](t2) at (-1,-5) {};
\node[rectangle,draw=black](t2) at (-1,-5) {$t_2$};
\node(dots2) at (0,-5) {$...$};
%\node[circle,draw=black, minimum size=4pt,inner sep=0pt, fill = black, label=below:{$t_{n+1}$}](tn1) at (1,-5) {};
\node[rectangle,draw=black](tn1) at (1,-5) {$t_{n+1}$};
%\node[circle,draw=black, minimum size=4pt,inner sep=0pt, fill = black, label=below:{$t_{n+2}$}](tn2) at (2,-5) {};
\node[rectangle,draw=black](tn2) at (2,-5) {$t_{n+2}$};

\draw[->,thick] (s) to (a10);
\draw[->,thick] (a10) to (dots);
%\draw [<-,thick](s) to [loop right, looseness = 80] (s);
\draw[->,thick] (dots) to (sn0);
\draw[->,thick] (sn0) to (sn);

%\draw[->,thick, dashed] (sn) to  (t1);
%\draw[->,thick, dashed] (sn) to (t2);
%\draw[->,thick,dashed] (sn) to (dots2);
%\draw[->,thick,dashed] (sn) to (tn1);
%\draw[->,thick,dashed] (sn) to (tn2);

\draw[->,thick] (sn) to  (t1);
\draw[->,thick] (sn) to (t2);
\draw[->,thick] (sn) to (dots2);
\draw[->,thick] (sn) to (tn1);
\draw[->,thick] (sn) to (tn2);
\end{tikzpicture}
\hspace{0.6cm}
   \begin{tikzpicture}[scale=0.7, transform shape]
%\node[circle,draw=black, minimum size=4pt,inner sep=0pt, fill = black, label=above:{$s_1$}](s) at (0,0) {};
\node[rectangle,draw=black](s) at (0,0) {$s_1$};
%\node[circle,draw=black, minimum size=4pt,inner sep=0pt,, fill = black , label=left:{$s_2$}](a10) at (0,-1) {};
\node[rectangle,draw=black](a10) at (0,-1) {$s_2$};
\node(dots) at (0,-2) {$...$};
%\node[circle,draw=black, minimum size=4pt,inner sep=0pt,, fill = black , label=left:{$s_n$}](sn0) at (0,-3) {};
\node[rectangle,draw=black](sn0) at (0,-3) {$s_n$};
%\node[circle,draw=black, minimum size=4pt,inner sep=0pt,, fill = black , label=left:{$s_{n+1}$}](sn) at (0,-4) {};
\node[rectangle,draw=black](sn) at (0,-4) {$s_{n+1}$};
%\node[circle,draw=black, minimum size=4pt,inner sep=0pt, fill = black, label=below:{$t_1: p$}](t1) at (-2,-5) {};
\node[rectangle,draw=black](t1) at (-2,-5) {$t_1: p$};
%\node[circle,draw=black, minimum size=4pt,inner sep=0pt, fill = black, label=below:{$t_2$}](t2) at (-1,-5) {};
\node[rectangle,draw=black](t2) at (-1,-5) {$t_2$};
\node(dots2) at (0,-5) {$...$};
%\node[circle,draw=black, minimum size=4pt,inner sep=0pt, fill = black, label=below:{$t_{n+1}$}](tn1) at (1,-5) {};
\node[rectangle,draw=black](tn1) at (1,-5) {$t_{n+2}$};
%\node[circle,draw=black, minimum size=4pt,inner sep=0pt, fill = black, label=below:{$t_{n+2}$}](tn2) at (2,-5) {};
\node[rectangle,draw=black](tn2) at (2,-5) {$t_{n+3}$};

\draw[->,thick] (s) to (a10);
\draw[->,thick] (a10) to (dots);
%\draw [<-,thick](s) to [loop right, looseness = 80] (s);
\draw[->,thick] (dots) to (sn0);
\draw[->,thick] (sn0) to (sn);

%\draw[->,thick, dashed] (sn) to  (t1);
%\draw[->,thick, dashed] (sn) to (t2);
%\draw[->,thick,dashed] (sn) to (dots2);
%\draw[->,thick,dashed] (sn) to (tn1);
%\draw[->,thick,dashed] (sn) to (tn2);

\draw[->,thick] (sn) to  (t1);
\draw[->,thick] (sn) to (t2);
\draw[->,thick] (sn) to (dots2);
\draw[->,thick] (sn) to (tn1);
\draw[->,thick] (sn) to (tn2);
\end{tikzpicture}
 
%}

\caption{Models $\M_{n+2}$ (left) and $\M_{n+3}$ (right) with reflexive arrows for $t$-states omitted. The cost of all edges is 1. }
\label{fig:pspaceex3}
\end{figure}
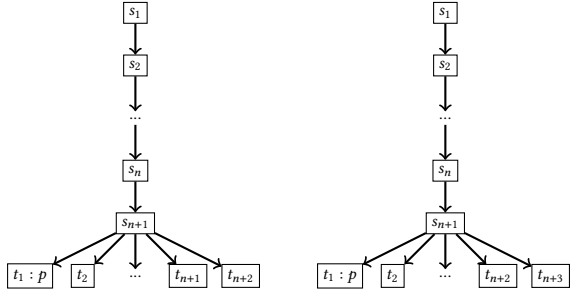 
The models are quite similar, with the only difference being that $\M_{n+2}$ has one less of $t$-states. For both models, $p$ is true only in state $t_1$. 

To show that $(\M_{n+2},s_1) \models \estrat 1 \mathsf{F} p$ and $(\M_{n+3},s_1) \not \models \estrat 1 \mathsf{F} p$, we argue that, while the traveller moves along the $s$-states, at each step in $\M_{n+2}$ the demon can remove one of the $s_{n+1} \to t_i$ edges with $i >1$. Such a gradual removal results in the situation, when in state $s_{n+1}$ the traveller has only one choice, namely to enter state $t_1$ thus satisfying $\mathsf{F} p$. In $\M_{n+3}$, we have one additional $t$-state, and therefore the demon does not have enough steps to ensure that only the $p$-state $t_1$ is available. Hence, the traveller can enter a  $t$-state that does not satisfy $p$ and violate $\mathsf{F} p$. Then we argue that $\varphi$ cannot distinguish the models, as in OL edges are restored after each game step, and hence the demon should make a move in state $s_{n+1}$ in one of the models that is not replicable in the other model. This is impossible, since the resource bound $n$ is too low.
\end{proof}

\textbf{The Expressivity Landscape.} An overview of the expressivity results is presented in Figure \ref{fig:expressivity}.

\begin{figure}[h!]
\centering
\begin{tikzpicture}[scale=0.7, transform shape]
\node (SDL) at (-4,-2) {SDL};
\node (ol) at (0,-2) {OL};
\node (SCL) at (4,-2) {SCL};
\node (sul) at (0,-5) {SUL};
\node (ctl) at (0,0.5) {CTL};

\draw[thick, ->] (ctl) to [bend right, looseness=0.8] node[sloped, anchor=center, above] {\small{Thm. \ref{thm:ctlexp}}}   (SDL);
\draw[thick, ->] (ctl) to  node[ anchor=center, left] {\small{\cite{CattaLM23}}}  (ol);
\draw[thick, ->] (ctl) to [bend left, looseness=0.8] node[sloped, anchor=center, above] {\small{Thm. \ref{thm:ctlexp}}} (SCL);

%\draw[thick, <->] (SDL) to node[anchor=center,  strike out,draw,-, label = below:{\footnotesize{Thm. \ref{thm:SCLvsSDL}}}] {} (SCL);

\draw[thick,->] (ol) to [bend left, looseness=0.5] node[sloped, anchor=center, below left] {\small{?}} (SDL);

\draw[thick,->] (SDL) to [bend left, looseness=0.5] node[sloped, anchor=center, strike out,draw,-, label = above:{\small{Thm. \ref{thm:SDLvsOL}}}] {} (ol);

\draw[thick, <->] (ol) to node[anchor=center,  strike out,draw,-, label = above:{\small{Thm. \ref{OLvsSCL}}}] {} (SCL);

\draw[thick, <->] (SDL) to [ looseness=0.7, out = 310, in = 230] node[anchor=center,  strike out,draw,-, label = above:{\small{Thm. \ref{thm:SCLvsSDL}}}] {} (SCL);

\draw[thick, ->] (SDL) to [bend right, looseness=0.8] node[sloped, anchor=center, below] {\small{Thm. \ref{thm:SCLvsSDL}}}   (sul);

\draw[thick, ->] (SCL) to [bend left, looseness=0.8] node[sloped, anchor=center, below] {\small{Thm. \ref{thm:SCLvsSDL}}}   (sul);

\draw[thick,->] (sul) to [looseness=0.7, out = 135, in = 225] node[sloped, near start, anchor=center, strike out,draw,-, label = below left:{\small{Thm. \ref{OLvsSCL}}}] {} (ol);

\draw[thick,->] (ol) to [looseness=0.7, out = 315, in = 45] node[near end, anchor=center, right] {\small{?}} (sul);

\end{tikzpicture}
\caption{The expressivity results. An arrow from $\mathcal{L}_1$ to $\mathcal{L}_2$ means $\mathcal{L}_1 \prec\mathcal{L}_2$. A strike-out arrow from $\mathcal{L}_1$ to $\mathcal{L}_2$ depicts $\mathcal{L}_1 \not \preccurlyeq \mathcal{L}_2$. Open problems are denoted with question marks.}
\label{fig:expressivity}
\end{figure}
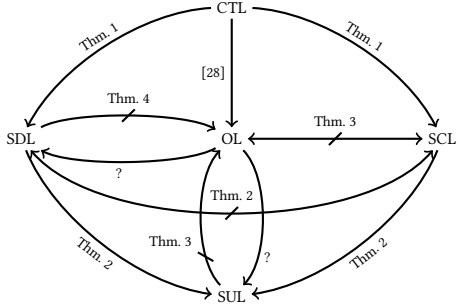

We leave the questions of whether OL $\not \preccurlyeq$ SDL and OL $\not \preccurlyeq$ SUL for future work and conjecture that it is indeed the case. In conclusion, we would also like to point out an interesting fact about the relationship between OL and SUL. Given  $(\M,s)$ and $\varphi$ of OL, if $(\M,s) \models \varphi$, then there is a formula $\psi$ of SUL s.t. $(\M,s)\models \psi$. The argument is directly based on the semantics of the two logics, with, e.g., $(\M,s) \models \langle \dagger^n \rangle \chi$ implying $(\M,s) \models \angeldemonstrat n n \chi$. The latter means that for each demonic strategy in OL, there is a joint demonic and angelic strategy in SUL over the same resource bounds reaching the same goal. Intuitively, to model the OL operator, the angel and the demon can cooperate where the demon removes edges, and the angel restores all or a subset of them in the next turn.

 \section{Model Checking}
 \label{sec:mc}

In this section, we study the model checking problem for SDL, SCL, and SUL. In particular, we show that for the first two logics, the problem is PSPACE-complete, and for SUL, the problem is in EXPSPACE. We also mention that the problem is PSPACE-complete for the next-time fragment of SUL. 

 \begin{definition}[Model Checking]
     Given a pointed model $(\M,s)$ and a formula $\varphi$, \emph{the model checking problem} consists in computing whether $(\M,s) \models \varphi$.
 \end{definition}

 Whenever necessary, we assume that $\varphi$ is in \textit{negation normal form} (NNF), meaning that in $\varphi$, negations only appear in front of atoms. It is straightforward to show that each formula $\varphi$ of SDL, SCL, and SUL, can be equivalently rewritten into $\varphi'$ in NNF using the propositional equivalences, duals (like $\astrat{n} \psi$ in the case of SDL), and temporal equivalences $\neg( \nextt \varphi) \leftrightarrow \nextt \neg \varphi$, $\neg (\varphi_1 \until \varphi_2) \leftrightarrow \neg \varphi_1 \release \neg \varphi_2$, and $\neg (\varphi_1 \release \varphi_2) \leftrightarrow \neg \varphi_1 \until \neg \varphi_2$. The size of the formula $\varphi'$ in NNF is at most linear in the size of the original $\varphi$.

 %We will also use the following abbreviations: $\square \varphi := \estrat 0 \nextt \varphi$ and $\Diamond \varphi:= \astrat 0 \nextt \varphi$. It is easy to verify that the box and diamond have exactly the standard modal logic semantics, as the only strategy the demon can play is keeping the model intact.
 %Observe that we can define the standard modal $\square$ and $\Diamond$ modalities using demonic strategies with 0 resources. Indeed, take $\square \varphi := \estrat 0 \nextt \varphi$ and $\Diamond \varphi:= \astrat 0 \nextt \varphi$. It is easy to verify that the box and diamond have exactly the intended semantics as the only strategy the demon can play is keeping the model intact.

\begin{restatable}{theorem}{sdlmc}
        Model checking SDL is PSPACE-complete. 
\end{restatable}

\begin{proof}
%The lower bound is shown via a reduction from the PSPACE-complete quantified Boolean formula (QBF) problem, and, for brevity, the proof is presented in Technical Appendix.

     We start with the lower bound and use a reduction from the PSPACE-complete quantified Boolean formula (QBF) problem. Given an instance of QBF $\Psi := Q_1p_1...Q_np_n\psi(p_1,...,p_n)$, where $Q_i \in \{\exists, \forall\}$, % and $p_i$'s are atoms, %for each $1\leq i \leq n$), 
     the problem consists in determining whether $\Psi$ is true. W.l.o.g., we assume that there are no free variables in $\Psi$ and that each variable is used for quantification only once.

    For a given instance of QBF $\Psi$, we construct a model $(\M^\Psi,s)$ %$ (\M,s),$ \munyquem{a cost function} $\C$, 
    and a formula $\varphi$ of SDL s.t. $ (\M^\Psi,s) \models \varphi$ iff $\Psi$ is true. Starting with the model, consider $ \M^\Psi = (S, \to, \mathcal{V}, \mathcal{C})$, where $S = \{s, s_1, ..., s_{2n}\}$, $s \to s$, $s \to s_i$ and $s_i \to s$ for all $s_i \in S$,  $\mathcal{V}(p_i^1) = \{s_i\}$ and $\mathcal{V}(p_i^0) = \{s_{n+i}\}$ for $1 \leqslant i \leqslant n$, and $\mathcal{C}(s,s) = 2$ and 1 for any other pair of states. %the cost of edges between any other pairs of states is 1. %For edges $s\to s$ and $s_i \to s$, we assign cost $2$, and for edges $s \to s_i$, the cost is 1. 
    Intuitively, for each $p_i$ in $\Psi$, we have two states with corresponding atoms $p_i^1$ and $p_i^0$, modelling whether $p_i$ was set to \textit{true} or \textit{false}.
    %\munyquem{Define the new propositional set for the constructed model, which has two propositions for each one considered in the QBF formula. }
   The model corresponding to a QBF instance $\Psi$ with two variables $p_1$ and $p_2$ is depicted in Figure \ref{fig:pspaceex1}.
     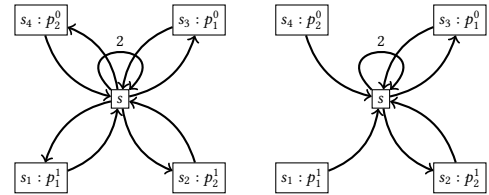
\begin{figure}[!ht]
 
\centering
%\scalebox{0.8}{
   \begin{tikzpicture}[scale=0.7, transform shape]
%\node[circle,draw=black, minimum size=4pt,inner sep=0pt, fill = black, label=left:{$s$}](s) at (0,0) {};
\node[rectangle,draw=black](s) at (0,0) {$s$};
%\node[circle,draw=black, minimum size=4pt,inner sep=0pt,, fill = black , label=below:{$s_1:p_1^1$}](a10) at (-2,-2) {};
\node[rectangle,draw=black](a10) at (-1.5,-1.5) {$s_1:p_1^1$};
%\node[circle,draw=black, minimum size=4pt,inner sep=0pt,, fill = black , label=below:{$s_2:p_2^1$}](a11) at (2,-2) {};
\node[rectangle,draw=black](a11) at (1.5,-1.5) {$s_2:p_2^1$};
%\node[circle,draw=black, minimum size=4pt,inner sep=0pt, , fill = black, label=above:{$s_4:p_2^0$}](a21) at (-2,2) {};
\node[rectangle,draw=black](a21) at (-1.5,1.5) {$s_4:p_2^0$};
%\node[circle,draw=black, minimum size=4pt,inner sep=0pt, , fill = black, label=above:{$s_3:p_1^0$}](a20) at (2,2) {};
\node[rectangle,draw=black](a20) at (1.5, 1.5) {$s_3:p_1^0$};

\draw[->,thick, bend right] (s) to  (a10);
\draw[->,thick,bend right] (a10) to  (s);
\draw[->,thick, bend right] (s) to  (a11);
\draw[->,thick,bend right] (a11) to  (s);
\draw[->,thick, bend right] (s) to  (a20);
\draw[->,thick,bend right] (a20) to  (s);
\draw[->,thick, bend right] (s) to  (a21);
\draw[->,thick,bend right] (a21) to  (s);
\draw [<-,thick](s) to [loop above, out=45, in=135, looseness = 10] node[above] {$2$} (s); %edited loop above right

\end{tikzpicture}
\hspace{0.5cm}
  \begin{tikzpicture}[scale=0.7, transform shape]
%\node[circle,draw=black, minimum size=4pt,inner sep=0pt, fill = black, label=left:{$s$}](s) at (0,0) {};
\node[rectangle,draw=black](s) at (0,0) {$s$};
%\node[circle,draw=black, minimum size=4pt,inner sep=0pt,, fill = black , label=below:{$s_1:p_1^1$}](a10) at (-2,-2) {};
\node[rectangle,draw=black](a10) at (-1.5,-1.5) {$s_1:p_1^1$};
%\node[circle,draw=black, minimum size=4pt,inner sep=0pt,, fill = black , label=below:{$s_2:p_2^1$}](a11) at (2,-2) {};
\node[rectangle,draw=black](a11) at (1.5,-1.5) {$s_2:p_2^1$};
%\node[circle,draw=black, minimum size=4pt,inner sep=0pt, , fill = black, label=above:{$s_4:p_2^0$}](a21) at (-2,2) {};
\node[rectangle,draw=black](a21) at (-1.5,1.5) {$s_4:p_2^0$};
%\node[circle,draw=black, minimum size=4pt,inner sep=0pt, , fill = black, label=above:{$s_3:p_1^0$}](a20) at (2,2) {};
\node[rectangle,draw=black](a20) at (1.5,1.5) {$s_3:p_1^0$};

%\draw[->,thick, bend right] (s) to  (a10);
\draw[->,thick,bend right] (a10) to  (s);
\draw[->,thick, bend right] (s) to  (a11);
\draw[->,thick,bend right] (a11) to  (s);
\draw[->,thick, bend right] (s) to  (a20);
\draw[->,thick,bend right] (a20) to  (s);
%\draw[->,thick, bend right] (s) to  (a21);
\draw[->,thick,bend right] (a21) to  (s);
\draw [<-,thick](s) to [loop above, out=45, in=135, looseness = 10] node[above] {$2$} (s); %edited loop above right

\end{tikzpicture}
\iffalse
\hspace{0.2cm}
   \begin{tikzpicture}[scale=0.7, transform shape]
%\node[circle,draw=black, minimum size=4pt,inner sep=0pt, fill = black, label=left:{$s$}](s) at (0,0) {};
\node[rectangle,draw=black](s) at (0,0) {$s$};
%\node[circle,draw=black, minimum size=4pt,inner sep=0pt,, fill = black , label=below:{$s_1:p_1^1$}](a10) at (-2,-2) {};
\node[rectangle,draw=black](a10) at (-2,-2) {$s_1:p_1^1$};
%\node[circle,draw=black, minimum size=4pt,inner sep=0pt,, fill = black , label=below:{$s_2:p_2^1$}](a11) at (2,-2) {};
\node[rectangle,draw=black](a11) at (2,-2) {$s_2:p_2^1$};
%\node[circle,draw=black, minimum size=4pt,inner sep=0pt, , fill = black, label=above:{$s_4:p_2^0$}](a21) at (-2,2) {};
\node[rectangle,draw=black](a21) at (-2,2) {$s_4:p_2^0$};
%\node[circle,draw=black, minimum size=4pt,inner sep=0pt, , fill = black, label=above:{$s_3:p_1^0$}](a20) at (2,2) {};
\node[rectangle,draw=black](a20) at (2,2) {$s_3:p_1^0$};

\draw[->,thick,bend right] (a10) to  (s);
\draw[->,thick,bend right] (a11) to   (s);
\draw[->,thick,bend right] (a20) to   (s);
\draw[->,thick,bend right] (a21) to  (s);
\draw [<-,thick](s) to [loop above, out=45, in=135, looseness = 10]  (s); %edited loop above right

\end{tikzpicture}
 \fi
 
%}

\caption{Model $\M^\Psi$ (left) and one of its updates (right). %for the QBF instance $\Psi$. 
Cost of all edges is 1 apart from the $s \to s$ edge that has cost 2.}
\label{fig:pspaceex1}
\end{figure}

Now, we will construct the corresponding formula $\varphi$ of SDL. First, recall that $n$ is the number of atoms in the QBF instance $\Psi$. Then, we define formula $chosen_k$ that will capture the fact that truth values of only the first $k$ atoms out of $n$ in $\Psi$ were set:
    $$chosen_k = \bigwedge_{1 \leqslant i \leqslant k} (\Diamond \Diamond p_i^0 \leftrightarrow \lnot \Diamond \Diamond p_i^1) \land \bigwedge_{k < i \leqslant n} (\Diamond \Diamond p^0_i \land \Diamond \Diamond p^1_i).$$
Recall that $\Diamond \varphi:= \astrat 0 \nextt \varphi$. Intuitively, $chosen_k$ holds if the demon has cut transitions to some of the first $k$ atoms $p^0_i$ and $p^1_i$ in such a way that if the transition to  $p^0_i$ is removed, then there must remain a transition to $p^1_i$. This simulates the unambiguous choice of the truth value of atom $p_i$ in $\Psi$. All other atoms beyond the first $k$ should still be accessible. We have double diamonds in the formula because each action of the demon is followed immediately by a move of the traveller. And if the traveller decides to go to one of the accessible $s_l$ states, we should still be able to verify the accessibility of some $p_i^{j \in \{0,1\}}$ in two steps passing $s$ along the way.  

%\munyque{ It is the first time the attacker is mentioned. It should be introduced in the definition.}

    Now, we are ready to tackle the construction of $\varphi$.
    \begin{align*}
    \varphi_0 &:= \psi(\Diamond \Diamond p^1_1, ..., \Diamond \Diamond p^1_n)\\
\varphi_k &:= 
\begin{cases} 
	\astrat 1 \nextt (chosen_k \to \varphi_{k-1}) &\text{if } Q_k = \forall\\
	\estrat 1 \nextt (chosen_k \land \varphi_{k-1}) &\text{if } Q_k = \exists\\
\end{cases}\\
\varphi &:= \varphi_{n}.
\end{align*}

%\munyque{ could we define the translation using a function so that the recursion is clearer? 
%}

What is left to show is that $\Psi:=Q_1 p_1 ... Q_n p_n \psi(p_1,...,p_n)$ is true if and only if $(\M^\Psi,s) \models \varphi$. First, observe that we consider demonic strategies with costs of up to 1, i.e., at each step the demon can remove up to one edge from $s$ to some $s_i$\footnote{Recall that by the definition of demonic strategies, the demon cannot remove an edge if it is the only outgoing edge from a given state.}. %This means that the demon can remove up to one edge in the model. 
Removing a transition to $p_i^0$ means that the truth value of $p_i$ is set to \textit{true}. For an example, see the model on the right in Figure \ref{fig:pspaceex1}, where the value of $p_1$ was set to \textit{false}, and $p_2$ was set to \textit{true}.  Guards $chosen_k$ ensure that truth-values of propositions are chosen unambiguously. Therefore, together with the guards, constructs $\astrat 1 \nextt$ and $\estrat 1 \nextt$ in the clause $\varphi_k$ of the translation emulate quantifiers $\forall$ and $\exists$. Once the truth values of all atoms $p_i$ were thus set, the evaluation of the QBF corresponds to the reachability of the corresponding atoms. In particular, we set $p_i$ to \textit{true} in an evaluation of $\Psi$ if and only if $p_i^1$ is reachable in the corresponding submodel.

%\munyque{ Perhaps recall that "strategies with costs of up to 1" means that the maximum cost of a node is one (to avoid confusion with the sum of edge costs). 
%}

%\munyque{w.l.o.g. }
To show that the model checking problem for SDL is in PSPACE, we present an alternating Algorithm \ref{quantMC}. Let $(\M,s)$ be a finite model, and $\varphi$ be a formula of SDL. Without loss of generality, we assume that formula $\varphi$ is in NNF. In the algorithm, we show only cases $\estrat n \nextt \psi$ and   $\estrat n  \psi_1 \until \psi_2$ for brevity, and the full algorithm can be found in the Technical Appendix.
\begin{breakablealgorithm}
	\caption{An algorithm for model checking SDL}\label{quantMC} 
	\small
 %\footnotesize
	\begin{algorithmic}[1] 		
		\Procedure{MC}{$(\M, s), \varphi$}		
        %\Case{$\varphi = p$}
        %\State{\textbf{return} $s \in \mathcal{V}(p)$}
        %\EndCase
        %\Case{$\varphi = \lnot p$}
        %\State{\textbf{return} not $s \in \mathcal{V}(p)$}
        %\EndCase
        %\Case{$\varphi = \psi \lor \chi$}
        %\State{\textbf{existentially choose} $\theta \in \{\psi, \chi\}$}
        %\State{\textbf{return} \textsc{MC}$((\M,s), \theta)$}
        %\EndCase
       %\Case{$\varphi = \psi \land \chi$}
        %\State{\textbf{universally choose} $\theta \in \{\psi, \chi\}$}
        %\State{\textbf{return} \textsc{MC}$((\M,s), \theta)$}
        %\EndCase
        \Case{$\varphi = \estrat n \nextt \psi$}
        \State{\textbf{existentially choose} $n$-submodel $\M'$}
        \State{\textbf{universally choose} $s'$ such that $s\xrightarrow{\M'} s'$}
        \State{\textbf{return} \textsc{MC}$((\M',s'), \psi)$}
        \EndCase
\Case {$\varphi = \estrat n  \psi_1 \until \psi_2$} 
\State{$X \gets (\M, s)$}
\State{$i \gets 0$}
\While{not \textsc{MC}$(X, \psi_2)$ and $i \leqslant brDepth$}
\If{not \textsc{MC}$(X, \psi_1)$}
\State{\textbf{return} \textit{false}}
\EndIf
\State{\textbf{existentially choose} $n$-submodel $\M'$ of $X$}
\State{\textbf{universally choose} $s'$ such that $s\xrightarrow{\M'} s'$}
\State{$X \gets (\M', s')$}

\State{$i \gets i+1$}
\EndWhile
\If{$i > brDepth$}
\State{\textbf{return} \textit{false}}
\Else
\State{\textbf{return} \textit{true}}
\EndIf

\EndCase

%\Case {$\varphi = \estrat n  \psi_1 \release \psi_2$} 
%\State{$X \gets (\M, s)$}
%\State{$i \gets 0$}
%\While{\textsc{MC}$(X,  \psi_2)$ and $i \leqslant brDepth$}
%\If{\textsc{MC}$(X,  \psi_1)$}
%\State{\textbf{return} \textit{true}}
%\EndIf
%\State{\textbf{existentially choose} $n$-submodel $\M'$}
%\State{\textbf{universally choose} $s'$ such that $s\xrightarrow{\M'} s'$}
%\State{$X \gets (\M', s')$}

%\State{$i \gets i+1$}
%\EndWhile
%\If{$i > brDepth$}
%\State{\textbf{return} \textit{true}}
%\Else
%\State{\textbf{return} \textit{false}}
%\EndIf

%\EndCase
   \EndProcedure

	\end{algorithmic}
\end{breakablealgorithm}

%In Algorithm \ref{quantMC}, we omit cases for $\astrat n \nextt \psi$, $\astrat n \psi_1 \until \psi_2$, and $\astrat n \psi_1 \release \psi_2$ for brevity. They are calculated similary to their diamond counterparts by switching existential and universal choices. 

In the algorithm, we explore the computational tree in a depth-first manner using universal and existential choices.
The algorithm follows closely the definition of the semantics, and the correctness can be shown by induction on $\varphi$. The termination follows from the fact that each case breaks down a subformula into simpler subformulas and the finite depth of the tree.

 The depth of each branch, $brDepth$, is at most $O(|\!\to\!|\cdot|S|)$, i.e., the demon has at most $|\!\to\!| + 1$ ways to consecutively modify the model (in the worst case, removing edges one by one, plus an option to not modify the model), and each such $n$-submodel has at most $|S|$ states where the traveller can transition to. In other words, for each $n$-submodel, it is enough to check up to $|S|$ states in the submodel. Thus, the total number of \textit{unique} game positions on a given branch, and hence the depth of each branch in the tree, is bounded by $O(|\!\to\!|\cdot|S|) \leqslant O(|\M|^2)$ and can be explored by our alternating algorithm in polynomial time. Since there are at most $|\varphi|$ subformulas to consider, the total running time of the algorithm is bounded by $O(|\varphi| \cdot |\M|^2)$. From the fact that APTIME (i.e., alternating polynomial time) = PSPACE \cite{chandra1981alternation}, we conclude that the model checking problem for SDL is PSPACE-complete.
\end{proof}

A knowledgeable reader may point out that while model checking SDL is PSPACE-complete, model checking its spiritual predecessor, OL, can be done in polynomial time \cite{CattaLM23}. This is due to the significant difference in the semantics of the two logics. In OL, edges are restored after the traveller has made a move, and in SDL, edge removal is permanent, i.e., we have to keep in memory a polynomial number of modified models to verify the strategic abilities of the demon. However, note that our PSPACE-completeness result is in line with the model checking results for logics with quantification over \textit{permanent} model change (see, e.g., \cite{loding2003model,DitmarschHKK17,FervariV19,galimullin24,delima14}). 

%Now we turn our attention to SCL. 

Model checking SCL is also PSPACE-complete by a relatively similar argument that we omit for brevity, and a proof sketch can be found in the Technical Appendix.

\begin{restatable}{theorem}{sclmc}
    Model checking SCL is PSPACE-complete.
\end{restatable}

Finally, we look into the complexity of the model checking problem for SUL and show that it is in EXPSPACE. However, we also note that the next-time fragment of SUL is still PSPACE-complete.

\begin{restatable}{theorem}{sulmc}
    Model checking SUL is in EXPSPACE, and model checking the next-time fragment of SUL is PSPACE-complete.
\end{restatable}

\begin{proof}
To show that model checking SUL is in EXPSPACE, we can provide an alternating algorithm similar to Algorithm \ref{quantMC}. Here, we present a general idea and the algorithm with the full argument can be found in Technical Appendix.

The main difference from Algorithm \ref{quantMC} is that for cases involving $\esetstrat n m$  we have to consider $n$-$m$-updates depending on $C$. Each $n$-$m$-update is of size bounded by $O(|\M|^2)$, and can be computed in polynomial time. For the next-time fragment of SUL, we need to check at most $|\varphi|$ such updates, and hence the total running time is in APTIME = PSPACE. The lower bound follows from the fact that SUL subsumes both SDL and SCL. 

For the full language, similarly to Algorithm \ref{quantMC}, we construct the computational tree with depth of branches bounded by $brDepth$. However, this time, $brDepth$ is exponential, since the demon and the angel together can, in the worst case, force an exponential number of submodels of the fully connected graph of size $|S|^2$, i.e. $brDepth$ is bounded by $O(2^{|\M|^2})$ and hence the algorithm runs in exponential time. Since AEXPTIME=EXPSPACE \cite{chandra1981alternation}, model checking SUL is in EXPSPACE.
\end{proof}

\section{Related Work}
\label{sec:rw}
\textbf{Obstruction Logics.} The initial inspiration for our work came from the research on \textit{Obstruction Logic} (OL) \cite{CattaLM23}, where one can reason about the strategies of the demon to deactivate some edges such that the target property holds regardless of the moves of the traveller. After each step of such a game, all edges are restored, which is different from our SDL setting, where we assume that edges are deactivated once and for all. Extensions of OL include obstruction ATL \cite{CattaLMM24}, timed OL \cite{LeneutreM025}, and coalition OL \cite{catta2025coalition}. All of these extensions feature only temporary edge deactivations. Compared to the body of research on OL, we have also introduced a formalism for reasoning about strategic \textit{activation} of edges, SCL, and also the interplay between the demon and the angel in SUL for scenarios of cooperation and competition. 

\textbf{Sabotage-like Logics.} The research on OL itself was motivated by sabotage games \cite{Benthem05}, where the demon can deactivate any one edge, and the corresponding \textit{Sabotage Modal Logic} (SML) \cite{loding2003model,aucher2018modal} and its generalisation to subsets of edges \cite{CattaLM23A}. Related to SML, where the demon can remove any edge in a model, there is also an approach for definable edge removal \cite{Li20}. %, where edges to be removed are further parametrised by a formula of a logic. 
In the same vein, there have also been logics for reasoning about adding an edge to a model, swapping two edges, copying and removing, etc (see, e.g., \cite{ArecesFH15,ArecesDFS17,ArecesDFMS21}). Compared to all these approaches, we consider weighted graphs, as well as extended strategies of the angel and the demon, as opposed to the next-time outcomes in the cited works. 

\textbf{Dynamic Epistemic Logic.} Adding and removing arrows in particular, and model updates in general, are a bread-and-butter in \textit{Dynamic Epistemic Logics} (DEL) \cite{hvdetal.del:2007}, that are built on epistemic models capturing knowledge of agents, and where updates of such models correspond to various information-changing events. %, like public learning, eavesdropping, suspicion, etc. 
Related to edge removal, one can mention arrow updates % and quantification over them 
\cite{KooiR11,DitmarschHKK17}, as well as various modes of agents sharing their knowledge with each other %that also result in edge removal (i.e., increase in knowledge) 
\cite{galimullin24,BaltagS20,AgotnesW17}. One can also think of adding arrows as an introspection effort of a given agent \cite{FervariV19}.
Moreover, there has been some work on incorporating strategic reasoning into DEL, in particular as reachability games over epistemic models \cite{MaubertPS19,maubert20}, concurrent public communication \cite{AgotnesD08,Galimullin21a}, and alternating-time DEL \cite{delima14}. Even though we can refer to DEL for some intuitions regarding model updates, needless to say, our setting is different from the one of agents' knowledge and learning.

\textbf{Dynamic Logics for Social Networks.} Inspired by DEL, there has been a considerable amount of research on dynamics in social networks (SNs) (see \cite{minathesis} for an overview). In the setting of SNs, adding or removing edges has been used, for example, to model changes in friendship \cite{Gonzalez22,edoardotark}, (un)following other agents \cite{XiongG19}, gaining knowledge in SNs \cite{ChristoffHP16}, visibility of posts on SNs \cite{GalimullinP24}, balance in a network \cite{DHoekKW20}, and sellers' strategies in diffusion auctions \cite{auctions}. 

\textbf{Strategic Reasoning.} In the realm of the strategic reasoning, we see model modifications, apart from the case of OL, primarily in \textit{normative reasoning}, where, given a MAS, a social law (or a norm), divides the set of actions of a given agent in a given state into desirable, or allowed, and undesirable, or prohibited. This is usually done by removing some of the transitions in a model (see, e.g., \cite{Alechina0P25,AlechinaG0P22,AlechinaLD18,bulling2016norm,agotnes10b,knobbout2016dynamic,GalimullinK24}). A general approach to modifications of strategic multi-agent models has recently been proposed in \cite{GalimullinGMM25}. Adding and removing edges was also interpreted as granting or revoking abilities to\textbackslash from agents \cite{galimullin21}. One can mention the research on module checking \cite{kupferman2001module,jamroga2015module}, where some transitions in an execution tree are cut to model various behaviours of the environment. Some types of model dynamics are prominent in separation logics (see, e.g., \cite{BednarczykDFM23,DemriD15,DemriF19,DBLP:conf/lics/Reynolds02}) that %were inspired by Hoare rules and that 
allow reasoning about the execution of computer programs. To the best of our knowledge, none of these approaches reason about extended temporal goals of agents that are able to modify the topology of the underlying model. 

\textbf{Logics with Resources.} In the context of epistemic logics, resources and costs were used to tackle the logical omniscience problem (see, e.g., \cite{FaginH87,Duc97,AlechinaL02,AlechinaL09}), as constraints on dynamic epistemic actions \cite{Costantini0P21,DolgorukovGG24,Solaki23}, and for epistemic planning \cite{BolanderDH21,BelardinelliR221}, to name a few use cases. In the context of strategic reasoning, there is a plethora of research on resource-bounded agents in the settings of CL and ATL (for example, see \cite{AlechinaLNR10,AlechinaLNR11,DCaoN17,DemriR23,BelardinelliD21,MonicaNP11,DBLP:conf/prima/CattaFM24}). Observe that in the latter case, models are \textit{static} as opposed to our \textit{dynamic} models, and hence the settings are very different. 

\textbf{Games on Graphs.} Finally, sabotage games are not the only games on graphs that involve a traveller and that were analysed with the tools of modal logic (see \cite{graphgames} for an overview). Such analyses include the logics for hide and seek \cite{LiGLT23}, cops and robbers \cite{pigsandrobbers}, and the poison games \cite{GrossiR19,poison}.

%Logics for removing arrows: sabotage \cite{loding2003model,aucher2018modal}, definable link removal \cite{Li20}, del\cite{hvdetal.del:2007} (arrow updates\cite{KooiR11,DitmarschHKK17}, copy and remove\cite{ArecesDFS17,ArecesDFMS21} sharing knowledge \cite{galimullin24,BaltagS20,AgotnesW17}), normative updates\cite{Alechina0P25,AlechinaG0P22,AlechinaLD18,bulling2016norm,agotnes10b}, obstruction \cite{CattaLM23}, subset sabotage logic\cite{CattaLM23A}

%References: Sabotage games \cite{Benthem05}; extensions of Obstruction Logic: OATL \cite{CattaLMM24}. Timed OTL \cite{LeneutreM025}. Coalition OTL \cite{catta2025coalition}. 

%Logics for adding arrows: bridge and in general \cite{ArecesFH15}, del, social networks (overview \cite{minathesis}, cutting friends links \cite{Gonzalez22} also here \cite{edoardotark}, (un)following \cite{XiongG19}, strategic balance in sn \cite{DHoekKW20}, epistemic learning in sn cutting \cite{ChristoffHP16} visibility\cite{GalimullinP24}), granting abilities \cite{galimullin21}
%Strategies over dynamic models: lamb  \cite{GalimullinGMM25}, reachability del\cite{MaubertPS19,maubert20}, atdel\cite{delima14}, coalition announcements \cite{AgotnesD08,Galimullin21a}, adding arrow as introspection in EL \cite{FervariV19}

%Games on graphs: overview \cite{graphgames}, cops and robbers \cite{pigsandrobbers}, hide and seek \cite{LiGLT23}, poison game \cite{GrossiR19,poison}

%Module checking\cite{kupferman2001module,jamroga2015module}, separation logic \cite{BednarczykDFM23,DemriD15,DemriF19,DBLP:conf/lics/Reynolds02}

\section{Discussion}
\label{sec:disc}
We have presented three novel logics for reasoning about strategies of agents that are able to modify a given model. In SDL, we have the demon who plays an edge-removing strategy such that for all paths taken by the traveller, a target condition holds. Similarly, to reason about edge-adding strategies of the angel, we have proposed SCL. Finally, to express the interaction between the demon and the angel in an ATL-like fashion, we introduced SUL. For all logics, we studied their expressivity and provided model-checking algorithms. 

Since we have just scratched the surface of reasoning about strategic model changes by agents, there is a plethora of future work. First, we would like to study the satisfiability problems of all three logics. Moreover, recall that we defined demonic and angelic strategies in a memory-less fashion, i.e., an action of an agent depends on the current pointed model. We would also like to study their perfect recall strategies, where their actions depend on histories consisting of pointed models. We conjecture that the semantics of all logics are equivalent for both types of strategies via an argument similar to the one for OL \cite{CattaLM23}. 

While defining SDL, we were inspired by OL. However, having introduced edge-adding strategies, it is very tempting to formalise and study a variant of OL, where instead of the demon we have an angel that adds edges for one turn only. Moreover, related to \cite{catta2025coalition}, we find it particularly intriguing to extend SUL to \textit{coalitions of angels and demons}. Finally, we would also like to consider extensions of our logics with greatest and least fixed points in the vein of \cite{Rohde06}.

%References
%%%%%%%%%%%%%%%%%%%%%%%%%%%%%%%%%%%%%%%%%%%%%%%%%%%%%%%%%%%%%%%%%%%%%%%%

%%% The acknowledgments section is defined using the "acks" environment
%%% (rather than an unnumbered section). The use of this environment 
%%% ensures the proper identification of the section in the article 
%%% metadata as well as the consistent spelling of the heading.

%\begin{acks}
%If you wish to include any acknowledgments in your paper (e.g., to 
%people or funding agencies), please do so using the `\texttt{acks}' 
%environment. Note that the text of your acknowledgments will be omitted
%if you compile your document with the `\texttt{anonymous}' option.
%\end{acks}

%%%%%%%%%%%%%%%%%%%%%%%%%%%%%%%%%%%%%%%%%%%%%%%%%%%%%%%%%%%%%%%%%%%%%%%%

%%% The next two lines define, first, the bibliography style to be 
%%% applied, and, second, the bibliography file to be used.

\bibliographystyle{ACM-Reference-Format} 
\bibliography{ref}

\clearpage
\appendix
\section*{Technical Appendix}

\sdlvsol*

\begin{proof}
Consider an SDL formula $\estrat 1 \mathsf{F} p$, and assume towards a contradiction that there is an equivalent formula $\varphi$ of OL. Since formulas of OL are finite, and each OL modality has a finite resource bound, we can assume that $n$ is the greatest $n \in \mathbb{N}$ appearing in $\varphi$. Now consider models $\M_{n+2}$ and $\M_{n+3}$ in Figure \ref{fig:pspaceex32} below.

    \begin{figure}[h!]
\centering
%\scalebox{0.8}{
   \begin{tikzpicture}[scale=0.7, transform shape]
%\node[circle,draw=black, minimum size=4pt,inner sep=0pt, fill = black, label=above:{$s_1$}](s) at (0,0) {};
\node[rectangle,draw=black](s) at (0,0) {$s_1$};
%\node[circle,draw=black, minimum size=4pt,inner sep=0pt,, fill = black , label=left:{$s_2$}](a10) at (0,-1) {};
\node[rectangle,draw=black](a10) at (0,-1) {$s_2$};
\node(dots) at (0,-2) {$...$};
%\node[circle,draw=black, minimum size=4pt,inner sep=0pt,, fill = black , label=left:{$s_n$}](sn0) at (0,-3) {};
\node[rectangle,draw=black](sn0) at (0,-3) {$s_n$};
%\node[circle,draw=black, minimum size=4pt,inner sep=0pt,, fill = black , label=left:{$s_{n+1}$}](sn) at (0,-4) {};
\node[rectangle,draw=black](sn) at (0,-4) {$s_{n+1}$};
%\node[circle,draw=black, minimum size=4pt,inner sep=0pt, fill = black, label=below:{$t_1: p$}](t1) at (-2,-5) {};
\node[rectangle,draw=black](t1) at (-2,-5) {$t_1: p$};
%\node[circle,draw=black, minimum size=4pt,inner sep=0pt, fill = black, label=below:{$t_2$}](t2) at (-1,-5) {};
\node[rectangle,draw=black](t2) at (-1,-5) {$t_2$};
\node(dots2) at (0,-5) {$...$};
%\node[circle,draw=black, minimum size=4pt,inner sep=0pt, fill = black, label=below:{$t_{n+1}$}](tn1) at (1,-5) {};
\node[rectangle,draw=black](tn1) at (1,-5) {$t_{n+1}$};
%\node[circle,draw=black, minimum size=4pt,inner sep=0pt, fill = black, label=below:{$t_{n+2}$}](tn2) at (2,-5) {};
\node[rectangle,draw=black](tn2) at (2,-5) {$t_{n+2}$};

\draw[->,thick] (s) to (a10);
\draw[->,thick] (a10) to (dots);
%\draw [<-,thick](s) to [loop right, looseness = 80] (s);
\draw[->,thick] (dots) to (sn0);
\draw[->,thick] (sn0) to (sn);

%\draw[->,thick, dashed] (sn) to  (t1);
%\draw[->,thick, dashed] (sn) to (t2);
%\draw[->,thick,dashed] (sn) to (dots2);
%\draw[->,thick,dashed] (sn) to (tn1);
%\draw[->,thick,dashed] (sn) to (tn2);

\draw[->,thick] (sn) to  (t1);
\draw[->,thick] (sn) to (t2);
\draw[->,thick] (sn) to (dots2);
\draw[->,thick] (sn) to (tn1);
\draw[->,thick] (sn) to (tn2);
\end{tikzpicture}
\hspace{0.6cm}
   \begin{tikzpicture}[scale=0.7, transform shape]
%\node[circle,draw=black, minimum size=4pt,inner sep=0pt, fill = black, label=above:{$s_1$}](s) at (0,0) {};
\node[rectangle,draw=black](s) at (0,0) {$s_1$};
%\node[circle,draw=black, minimum size=4pt,inner sep=0pt,, fill = black , label=left:{$s_2$}](a10) at (0,-1) {};
\node[rectangle,draw=black](a10) at (0,-1) {$s_2$};
\node(dots) at (0,-2) {$...$};
%\node[circle,draw=black, minimum size=4pt,inner sep=0pt,, fill = black , label=left:{$s_n$}](sn0) at (0,-3) {};
\node[rectangle,draw=black](sn0) at (0,-3) {$s_n$};
%\node[circle,draw=black, minimum size=4pt,inner sep=0pt,, fill = black , label=left:{$s_{n+1}$}](sn) at (0,-4) {};
\node[rectangle,draw=black](sn) at (0,-4) {$s_{n+1}$};
%\node[circle,draw=black, minimum size=4pt,inner sep=0pt, fill = black, label=below:{$t_1: p$}](t1) at (-2,-5) {};
\node[rectangle,draw=black](t1) at (-2,-5) {$t_1: p$};
%\node[circle,draw=black, minimum size=4pt,inner sep=0pt, fill = black, label=below:{$t_2$}](t2) at (-1,-5) {};
\node[rectangle,draw=black](t2) at (-1,-5) {$t_2$};
\node(dots2) at (0,-5) {$...$};
%\node[circle,draw=black, minimum size=4pt,inner sep=0pt, fill = black, label=below:{$t_{n+1}$}](tn1) at (1,-5) {};
\node[rectangle,draw=black](tn1) at (1,-5) {$t_{n+2}$};
%\node[circle,draw=black, minimum size=4pt,inner sep=0pt, fill = black, label=below:{$t_{n+2}$}](tn2) at (2,-5) {};
\node[rectangle,draw=black](tn2) at (2,-5) {$t_{n+3}$};

\draw[->,thick] (s) to (a10);
\draw[->,thick] (a10) to (dots);
%\draw [<-,thick](s) to [loop right, looseness = 80] (s);
\draw[->,thick] (dots) to (sn0);
\draw[->,thick] (sn0) to (sn);

%\draw[->,thick, dashed] (sn) to  (t1);
%\draw[->,thick, dashed] (sn) to (t2);
%\draw[->,thick,dashed] (sn) to (dots2);
%\draw[->,thick,dashed] (sn) to (tn1);
%\draw[->,thick,dashed] (sn) to (tn2);

\draw[->,thick] (sn) to  (t1);
\draw[->,thick] (sn) to (t2);
\draw[->,thick] (sn) to (dots2);
\draw[->,thick] (sn) to (tn1);
\draw[->,thick] (sn) to (tn2);
\end{tikzpicture}
 
%}

\caption{Models $\M_{n+2}$ (left) and $\M_{n+3}$ (right) with reflexive arrows for $t$-states omitted for readability. The cost of all edges is 1. }
\label{fig:pspaceex32}
\end{figure}
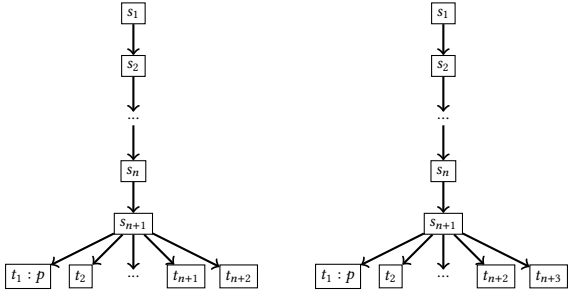 

The models are quite similar, with the only difference being that the one on the left has one less of $t$-states. For both models, $p$ is true only in state $t_1$. 

Now, to show that $(\M_{n+2},s_1)\models \varphi$ if and only if $(\M_{n+3},s_1) \models \varphi$, we use induction over $\varphi$, and we present a sketch of the argument here. For the base case, states $s_1$ in both models trivially satisfy the same propositional atoms. The Boolean cases follow as expected. Now, let us take a look at constructs $\langle \dagger^m \rangle \psi$ and $[ \dagger^m ]\psi$ with $m \leqslant n$. 
First, notice that the models are constructed in such a way that the demon cannot remove any edges between $s_i$-states. Moreover, recall that in OL, after a demon removes some edges and the traveller makes a move, all the edges are restored. Thus, it is enough to consider only demonic strategies from state $s_{n+1}$. Hence, what is left to argue is that for any demonic strategy in $(\M_{n+2}, s_{n+1})$ there is an equivalent demonic strategy in $(\M_{n+3}, s_{n+1})$, and vice versa. 

In the left-to-right direction, it is immediate that the choice of (up to $n$) $s_{n+1}-t_i$-edges to remove in $(\M_{n+2}, s_{n+1})$ can be directly replicated by exactly the same choice in $(\M_{n+3}, s_{n+1})$. Hence, if \textit{there exists} a demonic strategy in $(\M_{n+2}, s_{n+1})$, then there is an equivalent demonic strategy in  $(\M_{n+3}, s_{n+1})$. Moreover, \textit{any} demonic strategy in $(\M_{n+2}, s_{n+1})$ has an equivalent demonic strategy in $(\M_{n+3}, s_{n+1})$.

In the right-to-left direction, we can apply quite a similar argument with the only difference that if the demon removes the edge between $s_{n+1}$ and $t_{n+3}$ in $(\M_{n+3}, s_{n+1})$, then in $(\M_{n+2}, s_{n+1})$ we choose the first $t_i$-state such that $t_i \neq t_1$ and the edge from $s_{n+1}$ to $t_i$ is \textit{not} removed in $(\M_{n+3}, s_{n+1})$. Because we assume that $n$ is the greatest resource bound in $\varphi$, such a state $t_i$ will always be available. Since there $n+2$ and $n+3$ $t$-states correspondingly, the demon will always leave at least 2 (resp. 3) $t$-states accessible, i.e., at least one $t$-state not satisfying $p$ will always be accessible. In both cases, whether (i) $t_1$ is still accessible as well as $k$ of the other $t$-states, or (ii) there are only $k\leqslant n+2$ $t_i$-states with $i>1$ accessible, the traveller has to move into one of those states, where no further non-trivial demonic strategies are possible. It easy to see that $(\M_{n+2}, t_1)$ and $(\M_{n+3}, t_1)$ are not distinguishable by any OL formula, as well as $(\M_{n+2}, t_i)$ and $(\M_{n+3}, t_j)$ with $i,j > 1$.

Now, going back to our SDL formula $\estrat 1 \mathsf{F} p$, we can argue that $(\M_{n+2},s)\models \estrat 1 \mathsf{F} p$ and $(\M_{n+3},s)\not \models \estrat 1 \mathsf{F} p$. Let's start with $(\M_{n+2},s) \models \estrat 1 \mathsf{F} p$. By the definition of semantics, this is equivalent to the fact that there is a demonic strategy $\strat$ with cost $\mathcal{C}(\strat) \leqslant 1$ such that for all paths $\pi \in Out(\strat, (\M,s))$ we have that $\pi \models \mathsf{F} p$. A demonic strategy is a function from pointed models to sets of edges, and here we consider the strategy $\strat$ such that if given a pointed model with some $s_i$ state, then choose the edge $s_{n+1} \to t_{i+1}$ to remove, and if given a pointed model with some $t$-state, then choose the empty set of edges to remove. Intuitively, such a strategy will remove edges from $s_{n+1}$ to $t_i$ states, with $i>1$, one by one. This demonic strategy is played against all possible moves of the traveller. Notice that the model is constructed in such a way that the traveller does not have any other option than moving along the $s$-states. So by the time the traveller reaches state $s_{n+1}$, the demon will have removed $n$ edges from state $s_{n+1}$ to states from $t_2$ to $t_{n+1}$. In state  $s_{n+1}$ the demon, following her strategy, removes the edge between $s_{n+1}$ and $t_{n+2}$, and hence forces the traveller to enter the $t_1$ state, where $p$ is satisfied. 

To see that $(\M_{n+3},s) \not \models \estrat 1 \mathsf{F} p$, it is enough to notice that the same gradual edge removing trick will not work for model $\M_{n+3}$, as by the time the traveller reaches state $s_{n+1}$, the demon still have two edges to remove, from $s_{n+1}$ to $t_{n+2}$ and $t_{n+3}$, in order to force the satisfaction of $p$. This is impossible, as models are constructed in such a way that the demon can remove at most one edge per round. This implies that the demon cannot guarantee that the traveller ends up in state $t_1$ where $p$ holds.  
\end{proof}

\sdlmc*
\begin{proof}

The full algorithm for model checking SDL is presented below.
\begin{breakablealgorithm}
	\caption{An algorithm for model checking SDL}\label{quantMCapp} 
	\small
 %\footnotesize
	\begin{algorithmic}[1] 		
		\Procedure{MC}{$(\M, s), \varphi$}		
        \Case{$\varphi = p$}
        \State{\textbf{return} $s \in \mathcal{V}(p)$}
        \EndCase
        \Case{$\varphi = \lnot p$}
        \State{\textbf{return} not $s \in \mathcal{V}(p)$}
        \EndCase
        \Case{$\varphi = \psi \lor \chi$}
        \State{\textbf{existentially choose} $\theta \in \{\psi, \chi\}$}
        \State{\textbf{return} \textsc{MC}$((\M,s), \theta)$}
        \EndCase
       \Case{$\varphi = \psi \land \chi$}
        \State{\textbf{universally choose} $\theta \in \{\psi, \chi\}$}
        \State{\textbf{return} \textsc{MC}$((\M,s), \theta)$}
        \EndCase
        \Case{$\varphi = \estrat n \nextt \psi$}
        \State{\textbf{existentially choose} $n$-submodel $\M'$}
        \State{\textbf{universally choose} $s'$ such that $s\xrightarrow{\M'} s'$}
        \State{\textbf{return} \textsc{MC}$((\M',s'), \psi)$}
        \EndCase
\Case {$\varphi = \estrat n  \psi_1 \until \psi_2$} 
\State{$X \gets (\M, s)$}
\State{$i \gets 0$}
\While{not \textsc{MC}$(X, \psi_2)$ and $i \leqslant brDepth$}
\If{not \textsc{MC}$(X, \psi_1)$}
\State{\textbf{return} \textit{false}}
\EndIf
\State{\textbf{existentially choose} $n$-submodel $\M'$ of $X$}
\State{\textbf{universally choose} $s'$ such that $s\xrightarrow{\M'} s'$}
\State{$X \gets (\M', s')$}

\State{$i \gets i+1$}
\EndWhile
\If{$i > brDepth$}
\State{\textbf{return} \textit{false}}
\Else
\State{\textbf{return} \textit{true}}
\EndIf

\EndCase

\Case {$\varphi = \estrat n  \psi_1 \release \psi_2$} 
\State{$X \gets (\M, s)$}
\State{$i \gets 0$}
\While{\textsc{MC}$(X,  \psi_2)$ and $i \leqslant brDepth$}
\If{\textsc{MC}$(X,  \psi_1)$}
\State{\textbf{return} \textit{true}}
\EndIf
\State{\textbf{existentially choose} $n$-submodel $\M'$ of $X$}
\State{\textbf{universally choose} $s'$ such that $s\xrightarrow{\M'} s'$}
\State{$X \gets (\M', s')$}

\State{$i \gets i+1$}
\EndWhile
\If{$i > brDepth$}
\State{\textbf{return} \textit{true}}
\Else
\State{\textbf{return} \textit{false}}
\EndIf

\EndCase
   \EndProcedure

	\end{algorithmic}
\end{breakablealgorithm}

Cases for $\astrat n \nextt \psi$, $\astrat n \psi_1 \until \psi_2$, and $\astrat n \psi_1 \release \psi_2$ are calculated similarly to their diamond counterparts by switching existential and universal choices. 
\end{proof}

\sclmc*

\begin{proof}
    We can reuse Algorithm \ref{quantMCapp} for the argument for the PSPACE inclusion. In particular, we substitute every occurrence of $\estrat n$ and $\astrat n$ with $\angel n$ and $\allangel n$ correspondingly, and `$n$-submodel' with `$n$-supermodel'. The depth of each branch, $brDepth$, is at most $O(|S|^2 \cdot |S|)$ since the angel has at most $|S|^2$+1 ways to consecutively modify the model by adding edges one by one, and the traveller may be required to explore up to $|S|$ states of the resulting supermodels.

    For PSPACE-hardness, we again employ a reduction from the QBF problem. For a given instance of QBF $\Psi := Q_1p_1...Q_np_n\psi(p_1,..., \allowbreak p_n)$, we construct model $ \M^\Psi = (S, \to, \mathcal{V}, \mathcal{C})$, where $S = \{s, s_1, ..., \allowbreak s_{2n}\}$, $s \to s$ and $s_i \to s$ for all $s_i \in S$,  $\mathcal{V}(p_i^1) = \{s_i\}$ and $\mathcal{V}(p_i^0) = \{s_{n+i}\}$ for $1 \leqslant i \leqslant n$, and the cost of edges between any pairs of states is 1. An example of the model $\M^\Psi$ constructed for $\Psi:= \forall p_1 \exists p_2 (p_1 \to p_2)$ is presented in Figure \ref{fig:pspaceex}.

         \begin{figure}[!ht]
 
\centering
%\scalebox{0.8}{
\iffalse
   \begin{tikzpicture}[scale=0.7, transform shape]
%\node[circle,draw=black, minimum size=4pt,inner sep=0pt, fill = black, label=left:{$s$}](s) at (0,0) {};
\node[rectangle,draw=black](s) at (0,0) {$s$};
%\node[circle,draw=black, minimum size=4pt,inner sep=0pt,, fill = black , label=below:{$s_1:p_1^1$}](a10) at (-2,-2) {};
\node[rectangle,draw=black](a10) at (-2,-2) {$s_1:p_1^1$};
%\node[circle,draw=black, minimum size=4pt,inner sep=0pt,, fill = black , label=below:{$s_2:p_2^1$}](a11) at (2,-2) {};
\node[rectangle,draw=black](a11) at (2,-2) {$s_2:p_2^1$};
%\node[circle,draw=black, minimum size=4pt,inner sep=0pt, , fill = black, label=above:{$s_4:p_2^0$}](a21) at (-2,2) {};
\node[rectangle,draw=black](a21) at (-2,2) {$s_4:p_2^0$};
%\node[circle,draw=black, minimum size=4pt,inner sep=0pt, , fill = black, label=above:{$s_3:p_1^0$}](a20) at (2,2) {};
\node[rectangle,draw=black](a20) at (2,2) {$s_3:p_1^0$};

\draw[->,thick, bend right] (s) to  (a10);
\draw[->,thick,bend right] (a10) to  (s);
\draw[->,thick, bend right] (s) to  (a11);
\draw[->,thick,bend right] (a11) to  (s);
\draw[->,thick, bend right] (s) to  (a20);
\draw[->,thick,bend right] (a20) to  (s);
\draw[->,thick, bend right] (s) to  (a21);
\draw[->,thick,bend right] (a21) to  (s);
\draw [<-,thick](s) to [loop above, out=45, in=135, looseness = 10] node[above] {$2$} (s); %edited loop above right

\end{tikzpicture}
\hspace{0.2cm}
\fi
   \begin{tikzpicture}[scale=0.7, transform shape]
%\node[circle,draw=black, minimum size=4pt,inner sep=0pt, fill = black, label=left:{$s$}](s) at (0,0) {};
\node[rectangle,draw=black](s) at (0,0) {$s$};
%\node[circle,draw=black, minimum size=4pt,inner sep=0pt,, fill = black , label=below:{$s_1:p_1^1$}](a10) at (-2,-2) {};
\node[rectangle,draw=black](a10) at (-2,-2) {$s_1:p_1^1$};
%\node[circle,draw=black, minimum size=4pt,inner sep=0pt,, fill = black , label=below:{$s_2:p_2^1$}](a11) at (2,-2) {};
\node[rectangle,draw=black](a11) at (2,-2) {$s_2:p_2^1$};
%\node[circle,draw=black, minimum size=4pt,inner sep=0pt, , fill = black, label=above:{$s_4:p_2^0$}](a21) at (-2,2) {};
\node[rectangle,draw=black](a21) at (-2,2) {$s_4:p_2^0$};
%\node[circle,draw=black, minimum size=4pt,inner sep=0pt, , fill = black, label=above:{$s_3:p_1^0$}](a20) at (2,2) {};
\node[rectangle,draw=black](a20) at (2,2) {$s_3:p_1^0$};

\draw[->,thick,bend right] (a10) to  (s);
\draw[->,thick,bend right] (a11) to   (s);
\draw[->,thick,bend right] (a20) to   (s);
\draw[->,thick,bend right] (a21) to  (s);
\draw [<-,thick](s) to [loop above, out=45, in=135, looseness = 10]  (s); %edited loop above right

\end{tikzpicture}

%}

\caption{Model $\M^\Psi$ for the QBF instance $\Psi$. Cost of all edges is 1.}
\label{fig:pspaceex}
\end{figure}
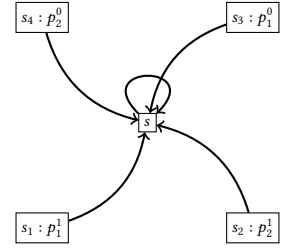 

    Similarly to the case of SDL, we define the standard modal diamond as an abbreviation $\Diamond \varphi := \allangel 0 \nextt \varphi$. The corresponding formula $chosen_k$ is then defined as:
    $$chosen_k = \bigwedge_{1 \leqslant i \leqslant k} (\Diamond \Diamond p_i^0 \leftrightarrow \lnot \Diamond \Diamond p_i^1) \land \bigwedge_{k < i \leqslant n} ( \lnot \Diamond \Diamond p^0_i \land \lnot \Diamond \Diamond p^1_i).$$

Formula $chosen_k$ states that the truth-values of the first $k$ atoms in $\Psi$ were set unambiguously (i.e., for each $p_i$ for $1 \leqslant i \leqslant k$ exactly one of $p_i^0$ and $p_i^1$ is reachable), and the truth-values of other atoms have not been set (i.e., for each $p_i$ with $k < i \leqslant n$ neither $p_i^0$ nor $p_i^1$ is reachable). The main difference from the case of SDL is that now we add edges instead of removing them. 

    The construction of $\varphi$ is then as follows:
    \begin{align*}
    \varphi_0 &:= \psi(\Diamond \Diamond p^1_1, ..., \Diamond \Diamond p^1_n)\\
\varphi_k &:= 
\begin{cases} 
	\allangel 1   \nextt (chosen_k \to \varphi_{k-1}) &\text{if } Q_k = \forall\\
	 \angel 1   \nextt (chosen_k \land \varphi_{k-1}) &\text{if } Q_k = \exists\\
\end{cases}\\
\varphi &:= \varphi_{n}.
\end{align*}
The rest of the argument is similar to the one for the PSPACE-hardness of SDL.
\end{proof}

\sulmc*

\begin{proof}
To show that model checking SUL is in EXPSPACE, we present an alternating Algorithm \ref{quantMC2}, which is similar to Algorithm \ref{quantMC}. For brevity, we consider only cases $\esetstrat n m \nextt
 \psi$ and $\esetstrat n m \psi_1 \until \psi_2$ and other cases can be computed similarly. W.l.o.g., we assume that formula $\varphi$ is in NNF. Given a finite model $\M = (S, \to, \mathcal{V}, \mathcal{C})$,  
 $n$-supermodel $\M_1 = \M \cup A$, and $m$-submodel $\M_2 = \M \setminus B$, where $A$ and $B$ are subsets of the set of edges, the size of the $n$-$m$-update $\M_1 \star \M_2$ is at most $O(|\M|^2)$ (i.e. fully connected graph).
 %we will denote by $\M_1 \star \M_2 = (\M \setminus A) \cup B$ the resulting model after the demon and the angel has removed or added their sets of edges. Note that the sets of edges the angel and the demon manipulate are disjoint, and the size of $\M_1 \star \M_2$ is at most $O(|\M|^2)$ (i.e. fully connected graph).
\begin{breakablealgorithm}
	\caption{An algorithm for model checking SUL}\label{quantMC2} 
	\small
 %\footnotesize
	\begin{algorithmic}[1] 		
		\Procedure{MC}{$(\M, s), \varphi$}		
        
        \Case{$\varphi = \esetstrat n m \nextt \psi$}
        
        \State{\textbf{return} \textsc{MC}$(\textsc{Update} ((\M, s),C, n, m), \psi)$}
        \EndCase
\Case {$\varphi = \esetstrat n m  \psi_1 \until \psi_2$} 
\State{$X \gets (\M, s)$}
\State{$i \gets 0$}
\While{not \textsc{MC}$(X, \psi_2)$ and $i \leqslant brDepth$}
\If{not \textsc{MC}$(X, \psi_1)$}
\State{\textbf{return} \textit{false}}
\EndIf
%\State{\textbf{existentially choose} $n$-submodel $\M'$}
%\State{\textbf{universally choose} $s'$ such that $s\xrightarrow{\M'} s'$}
\State{$X \gets \textsc{Update}(X, C, n, m)$}

\State{$i \gets i+1$}
\EndWhile
\If{$i > brDepth$}
\State{\textbf{return} \textit{false}}
\Else
\State{\textbf{return} \textit{true}}
\EndIf

\EndCase

   \EndProcedure

   \Procedure{Update}{$(\M,s), C, n, m$}
           \State{$X \gets (\M, s)$}
        \If{$C = \{\pentacle^m\}$}
            \State{\textbf{existentially choose} $m$-submodel $\M_1$ of $\M$}
            \State{\textbf{univ. choose} $n$-supermodel $\M_2$  of $\M$ \textbf{and} $s'$ s.t. $s\xrightarrow{\M_1 \star \M_2} s'$}
            %\State{$X \gets (\M_1 \star \M_2, s')$}
            \State{\textbf{return} $(\M_1 \star \M_2, s')$}
        \ElsIf{$C = \{\angelsymbol^n\}$}
         \State{\textbf{existentially choose} $n$-supermodel $\M_1$  of $\M$}
            \State{\textbf{univ. choose} $m$-submodel $\M_2$  of $\M$ \textbf{and} $s'$ s.t. $s\xrightarrow{\M_1 \star \M_2} s'$}
            %\State{$X \gets (\M_1 \star \M_2, s')$}
            \State{\textbf{return} $(\M_1 \star \M_2, s')$}
        \ElsIf{$C = \{\angelsymbol^n, \pentacle^m\}$}
         \State{\textbf{existentially choose} $n$-$m$-update $\M'$  of $\M$}
            \State{\textbf{universally choose} $s'$ s.t. $s\xrightarrow{\M'} s'$}
            %\State{$X \gets (\M', s')$}
            \State{\textbf{return} $(\M', s')$}
        \Else
        \State{\textbf{universally choose} $n$-$m$-update $\M'$  of $\M$ \textbf{and} $s'$ s.t. $s\xrightarrow{\M'} s'$}
        %\State{$X \gets (\M', s')$}
        \State{\textbf{return} $(\M', s')$}
        \EndIf
        %\State{\textbf{return} $X$}
        \EndProcedure

	\end{algorithmic}
\end{breakablealgorithm}

In the algorithm, procedure \textsc{Update} determines the updated model to be used based on the given model and the set $C$. Regarding complexity, procedure \textsc{Update} runs in polynomial time. For the next-time fragment of SUL, Algorithm \ref{quantMC2} checks at most $|\varphi|$ subformulas of $\varphi$, and hence the total running time is polynomial. From APTIME=PSPACE, we get that the next-time fragment of SUL is in PSPACE. Hardness follows trivially from the fact that both SDL and SCL are subsumed by SUL, and both logics have PSPACE-complete model checking problems.

Now, let us turn to the full language of SUL. The algorithm is quite similar to the one for SDL and SCL. However, the depth of each branch in our path tree, $brDepth$, is now exponential, since the demon and the angel together can, in the worst case, force an exponential number of submodels of the fully connected graph of size $|S|^2$ to consider. For each such submodel, we also may have to check up to $|S|$ states the traveller can access. In particular, $brDepth$ can be at most $O(2^{|S|^2}\cdot |S|) \leqslant O(2^{|\M|^2})$. Since there are at most $|\varphi|$ subformulas to consider, Algorithm \ref{quantMC2} runs in time bounded by $O(|\varphi| \cdot 2^{|\M|^2})$. From the fact that AEXPTIME = EXPSPACE \cite{chandra1981alternation}, we 
conlude that model checking SUL is in EXPSPACE.
\end{proof}

%%%%%%%%%%%%%%%%%%%%%%%%%%%%%%%%%%%%%%%%%%%%%%%%%%%%%%%%%%%%%%%%%%%%%%%%

\end{document}